\documentclass[a4paper,11pt]{article}
\usepackage{lmodern}
\usepackage[utf8]{inputenc}
\usepackage[T1]{fontenc}
\usepackage[USenglish]{babel}
\usepackage{amsmath,amssymb,amsthm}
\usepackage{stmaryrd}

\usepackage{fullpage}

\usepackage{graphicx}
\usepackage{xspace}
\usepackage{enumerate}
\usepackage[labelsep=period,labelfont={sc},textfont=it,margin=\parindent,font=small,format=plain]{caption}
\usepackage[numbers,sort&compress]{natbib}
\usepackage[%
  pdftitle={Minimal and minimum unit circular-arc models},%
  pdfauthor={Francisco J.\ Soulignac and Pablo Terlisky},%
  pdfcreator={},%
  pdfsubject={},%
  pdfkeywords={},%
  colorlinks=true,%
  linkcolor=blue,%
  citecolor=blue,%
  urlcolor=blue]{hyperref}  
\usepackage{doi}
  
\newtheorem{theorem}{Theorem}
\newtheorem{lemma}[theorem]{Lemma}
\newtheorem{observation}[theorem]{Observation}
\newtheorem{corollary}[theorem]{Corollary}

\newcommand{\Figure}{Fig.}
\newcommand{\Figures}{Figs.}

\newcommand{\NP}{\ensuremath{\text{\bf NP}}}
\newcommand{\A}{\ensuremath{\mathcal{A}}}

\newcommand{\C}{\ensuremath{\mathcal{C}}}

\renewcommand{\L}{\ensuremath{\mathcal{L}}}
\newcommand{\M}{\ensuremath{\mathcal{M}}}

\renewcommand{\S}{\ensuremath{\mathcal{S}}}
\newcommand{\T}{\ensuremath{\mathcal{T}}}
\newcommand{\U}{\ensuremath{\mathcal{U}}}

\newcommand{\W}{\ensuremath{\mathcal{W}}}

\DeclareMathOperator{\Ext}{ext}
\DeclareMathOperator{\Jump}{jmp}
\DeclareMathOperator{\Sep}{sep}
\DeclareMathOperator{\Gr}{Gr}

\DeclareMathOperator{\Row}{row}
\DeclareMathOperator{\Ind}{ind}
\DeclareMathOperator{\Rows}{rows}
\DeclareMathOperator{\Col}{col}
\DeclareMathOperator{\Pos}{pos}
\newcommand{\Unit}{\U}
\newcommand{\Desc}{u}
\newcommand{\Circ}{c}
\newcommand{\Len}{\ell}
\newcommand{\Dist}{d}
\newcommand{\BegDist}{d_s}
\newcommand{\Syn}{\S}
\newcommand{\SynInt}{\L}

\newcommand{\Copies}{\kappa}
\newcommand{\Unroll}{\cdot}
\newcommand{\Noses}{\nu}
\newcommand{\Hollows}{\eta}
\newcommand{\Steps}{\sigma}

\newcommand{\Range}[1]{\llbracket #1\rrbracket}

\DeclareMathOperator{\Wg}{wg}

\begin{document}

\title{Minimal and minimum unit circular-arc models}

\author{Francisco J.\ Soulignac\thanks{CONICET}~\thanks{Departamento de Ciencia y Tecnología, Universidad Nacional de Quilmes, Bernal, Argentina.} \and Pablo Terlisky\footnotemark[2]~\thanks{Instituto de Cálculo, FCEN, Universidad de Buenos Aires, Buenos Aires, Argentina.}}

\date{\normalsize\texttt{francisco.soulignac@unq.edu.ar}, \texttt{terlisky@dc.uba.ar}}

\maketitle

\begin{abstract}
  A \emph{proper circular-arc (PCA) model} is a pair $\M = (C, \A)$ where $C$ is a circle and $\A$ is a family of inclusion-free arcs on $C$ in which no two arcs of $\A$ cover $C$.  A PCA model $\Unit = (C,\A)$ is a \emph{$(\Circ, \Len)$-CA} model when $C$ has circumference $\Circ$, all the arcs in $\A$ have length $\Len$, and all the extremes of the arcs in $\A$ are at a distance at least $1$.  If $\Circ \leq \Circ'$ and $\Len \leq \Len'$ for every $(\Circ', \Len')$-CA model equivalent (resp.\ isomorphic) to $\Unit$, then $\Unit$ is \emph{minimal} (resp.\ \emph{minimum}).  In this article we prove that every PCA model is isomorphic to a minimum model.  Our main tool is a new characterization of those PCA models that are equivalent to $(\Circ,\Len)$-CA models, that allows us to conclude that $\Circ$ and $\Len$ are integer when $\Unit$ is minimal.  As a consequence, we obtain an $O(n^3)$ time and $O(n^2)$ space algorithm to solve the minimal representation problem, while we prove that the minimum representation problem is \NP-complete.
  
\end{abstract}

\section{Introduction}

The last decade saw an increasing research on numerical problems for \emph{unit interval (UIG)} and \emph{unit circular-arc (UCA)} models~\cite{CostaDantasSankoffXuJBCS2012,DuranFernandezGrippoSouzaSzwarcfiterENiDM2015,KlavikKratochvilOtachiRutterSaitohSaumellVyskocilA2017,LinSoulignacSzwarcfiter2009,LinSzwarcfiterSJDM2008,SoulignacJGAA2017,SoulignacJGAA2017a}.  In these problems we are given a UCA (or UIG) model $\M$ and we have to find UCA (or UIG) model $\Unit$, related to $\M$, that satisfies certain numerical constraints.  Here we consider two numerical problems, whose constraints ask to minimize the circumference of the circle and lengths of the arcs of $\Unit$.  To define these problems, we require some terminology that will be used in the remaining of this article. 

\paragraph{Statement of the problems.}

A \emph{proper circular-arc} (PCA) model $\M$ is a pair $(C, \A)$, where $C$ is a circle and $\A$ is a finite family of inclusion-free arcs of $C$ in which no pair of arcs in $\A$ cover $C$.  If $s, t$ are points of $C$, then $(s, t)$ is the open arc of $C$ that goes from $s$ to $t$ in a clockwise traversal of $C$, while $|s,t|$ is the length of $(s,t)$.  Each arc $A = (s,t) \in \A$ is described by its \emph{extremes} $s(A) = s$ and $t(A) = t$.  The \emph{extremes} of $\M$ are those extremes of the arcs in $\A$.  An ordered pair of extremes $e_1e_2$ of $\M$ is \emph{consecutive} when $\M$ has no extremes in $(e_1, e_2)$.  We assume $C$ has a special point $0$ such that $p = |0,p|$ for every point $p \in C$.  We classify the arcs of $\A$ as being \emph{external} or \emph{internal} according to whether $A \cup \{t(A)\}$ contains $0$ or not, respectively.  For $A_1, A_2 \in \A$, we write $A_1 < A_2$ to mean that $s(A_1)$ appears before $s(A_2)$ in a clockwise traversal of $C$ from $0$.  

A \emph{unit circular-arc} (UCA) model is a PCA model $\M$ whose arcs all have the same length $\Len$.  If $|e,e'| \geq 1$ for every pair of consecutive extremes $e$ and $e'$, then we refer to $\M$ as being a \emph{$(|C|, \Len)$-CA} model.  In this work, \emph{proper interval (PIG)} and \emph{unit interval (UIG)} models correspond to those PCA and UCA models that have no external arcs, respectively.

Two ingredients are required to define a numerical problem: the relation between the input and output, and the numerical constraint.  Equivalence and isomorphism are the relations that we consider in this article (\Figure~\ref{fig:models}).  Two PCA models $\M = (C, \A)$ and $\M' = (C', \A')$ are \emph{equivalent} when their extremes appear in the same order in the traversals of $C$ and $C'$ from their respective $0$ points, while $\M$ and $\M'$ are \emph{isomorphic} when the intersection graphs of $\A$ and $\A'$ are isomorphic.  Formally, $\M$ and $\M'$ are equivalent (resp.\ isomorphic) if there exists a bijection $f\colon\A \to \A'$ such that $e(f(A)) < e'(f(B))$ (resp.\ $f(A) \cap f(B) \neq \emptyset$) if and only if $e(A) < e'(B)$ (resp.\ $A \cap B \neq \emptyset$), for $e, e' \in \{s,t\}$. 

\begin{figure}[t!]
  \mbox{} \hfill \includegraphics{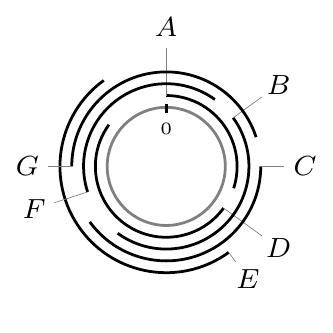} \hfill \includegraphics{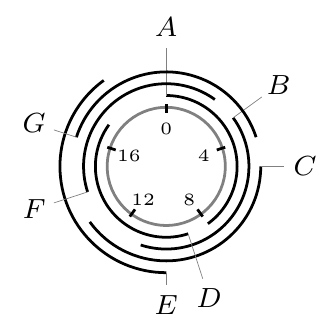} \hfill \includegraphics{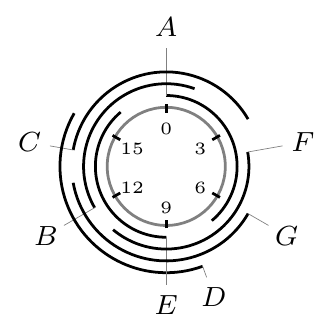} \hfill\mbox{} 
  \caption{From left to right: a PCA model $\M$, a minimal $(20,8)$-CA model equivalent to $\M$, and a minimum $(18,7)$-CA model isomorphic to $\M$.}\label{fig:models}
\end{figure}

Regarding the numerical constraints, our goal is to \textbf{simultaneously} minimize the circumference of the circle and the length of the arcs (\Figure~\ref{fig:models}).  Let $\Unit$ be a $(\Circ, \Len)$-CA model.  Formally, $\Unit$ is \emph{minimal} (resp.\ \emph{minimum}) when $\Len \leq \Len'$ and $\Circ \leq \Circ'$ for every $(\Circ', \Len')$-CA model equivalent (resp.\ isomorphic) to $\Unit$.  The \emph{minimal} (resp.\ \emph{minimum}) \emph{representation problem} asks to find a minimal (resp.\ minimum) UCA model equivalent (resp.\ isomorphic) to an input UCA model.

\paragraph{Brief history of the problems.}

As described in~\cite{SoulignacJGAA2017}, the motivations to study the minimal representation problems, both on UIG and UCA models, date back to c.~1950 at least, and thus predate the notions of UIG and UCA graphs.  The formal definition of minimal models appeared in 1990, when Pirlot~\cite{PirlotTaD1990} proved that every UIG model $\M$ is equivalent to some minimal $(\Circ, \Len)$-CA model, and that $\Circ$ and $\Len$ are integer values.  Pirlot's work yields an $O(n^2)$ time algorithm to decide if $\M$ is equivalent to a $(\Circ, \Len)$-CA model, when $\Circ$ and $\Len$ are given as input, that can be used to compute a minimal representation in $O(n^2\log n)$ time.

The main tool devised by Pirlot is a new representation of PIG models, called \emph{synthetic} graphs.  As observed by Mitas~\cite{Mitas1994}, synthetic graphs admit peculiar plane drawings that provide a framework to prove different properties with simple geometrical arguments.  In particular, Mitas uses these drawings to solve the minimal representation problem in $O(n^2)$ time and $O(n)$ space.\footnote{\textbf{N.B.} Even though Mitas' original algorithm is linear, it has a mistake and the correct version requires quadratic time~\cite{SoulignacJGAA2017a}.} The textbook~\cite{PirlotVincke1997} devotes a chapter to Pirlot's and Mitas' works, while it provides other reasons for studying the minimal representation problem.  

Recently, Klav\'\i{}k et al.~\cite{KlavikKratochvilOtachiRutterSaitohSaumellVyskocilA2017} rediscovered synthetic graphs while dealing with the \emph{bounded representation problem} on UIG models, while Soulignac~\cite{SoulignacJGAA2017,SoulignacJGAA2017a} extended synthetic graphs to UCA models.  As part of his work, Soulignac proves that every UCA model $\M$ is equivalent to some minimal $(\Circ, \Len)$-CA model $\Unit$ that, under the unproved assumption that $\Circ$ and $\Len$ are integer values, can be computed in $O(n^4\log n)$ time.  Besides conjecturing that $\Circ$ and $\Len$ must be integer, Soulignac asks for an efficient algorithm to solve the minimum representation problem, strengthening the open problems of \textbf{independently} computing the minimum arc and circle lengths of a UCA graph, reported by Lin and Szwarcfiter~\cite{LinSzwarcfiterSJDM2008}.  As noted by Soulignac, the minimum and minimal problems coincide for UIG models, but they differ in the UCA case~(\Figure~\ref{fig:models}).

We refer to~\cite{SoulignacJGAA2017,SoulignacJGAA2017a} for a deeper and up-to-date overview of these and other representation problems on UIG and UCA models.

\paragraph{Our contributions.}

\begin{figure}[t!]
  \centering
  \begin{tabular}{c@{\hspace{1.2cm}}c@{\hspace{1.2cm}}c}
    \includegraphics{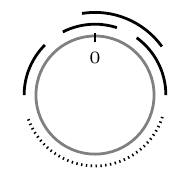} & \includegraphics{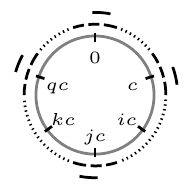} & \includegraphics{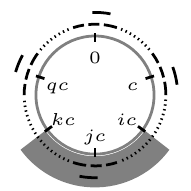} \\
    $\M$ & $\Copies \Unroll \M$ & Segment of $\Copies \Unroll \M$
  \end{tabular}
  \caption{The loop unrolling technique; $\kappa' = \kappa-1$.}\label{fig:model unrolling}
\end{figure}

We prove that $\Circ$ and $\Len$ are integer values when $\M$ is a minimal $(\Circ, \Len)$-CA model, and that every UCA model is isomorphic to some minimum UCA model.  Then, we devise an $O(n^3)$ time and $O(n^2)$ space algorithm for the minimal representation problem, while we prove that the minimum representation problem is \NP-complete.

From a theoretical point of view, our main contribution is a new characterization of the family $\mathbb{M}$ of PCA models that have equivalent UCA models.  As discussed in Section~\ref{sec:characterization}, our characterization simplifies the criterion for recognizing if $\M \in \mathbb{M}$ given by Tucker~\cite{TuckerDM1974}.  From a technical point of view, we apply a simple ``loop unrolling'' technique similar to that used for computer programs (\Figure~\ref{fig:model unrolling}).  Loosely speaking, we replicate $\Copies$ times the arcs of a PCA model $\M$.  As it turns out, we can determine if $\M \in \mathbb{M}$ by looking only at a segment of the model $\Copies \Unroll\M$ so obtained.  Moreover, the information in this segment, which is a UIG model, is enough to determine the minimum $\Circ$ and $\Len$ for which $\M \in \mathbb{M}$ is equivalent to a $(\Circ, \Len)$-CA model.  Loop unrolling is a natural and old technique that, not surprisingly, has already been applied to circular-arc models (e.g.~\cite{WerraEisenbeisLelaitMarmolDAM1999}).

\subsection{Preliminaries}

This section describes the remaining non-standard definitions that we use throughout the article.  

For $m < n$, we write $\Range{m,n} = [m,n) \cap \mathbb{N}$ and $\Range{n} = \Range{0,n}$.  When $S$ is a set with $|S| = n$, we use $\Range{S}$ to denote $\Range{n}$.  

A \emph{$q$-digraph} is a $(q+1)$-tuple $D = (V, E_0, \ldots, E_{q-1})$ such that $(V, E_i)$ is a digraph that can contain loops but not multiple edges, for $i \in \Range{q}$.  We write $V(D) = V$ and $E(D) = \bigcup_{i\in\Range{q}} E_i$ to denote the set of \emph{vertices} and bag of \emph{edges} of $D$, respectively, and $n = |V(D)|$ and $m = |E(D)|$.  For any pair $u,v \in V(D)$, we interchangeably write $uv$ or $u \to v$ to denote the ordered pair $(u,v)$.  In some occasions we may refer to $uv$ as being an edge \emph{from} (resp.\ \emph{starting at}) $u$ \emph{to} (resp.\ \emph{ending at}) $v$, regardless of whether $uv \in E(D)$.  

A walk $W$ in a $q$-digraph $D$ is a sequence of edges $v_0v_1, v_1v_2 \ldots, v_{k-1}v_k$ of $G$; walk $W$ goes \emph{from} (or \emph{begins at}) $v_0$ \emph{to} (or \emph{ends at}) $v_k$.  We say that $W$ is a \emph{circuit} when $v_k = v_0$, that $W$ is a \emph{path} when $v_i \neq v_j$ for every $0 \leq i < j \leq k$, and that $W$ is a \emph{cycle} when it is a circuit and $v_0v_1, \ldots, v_{k-2}v_{k-1}$ is a path.  If $D$ contains no cycles, then $D$ is \emph{acyclic}.  For the sake of notation, we could say that $W$ is a \emph{circuit} when $v_0 \neq v_k$; this means that $W, v_kv_0$ is a circuit.  Moreover, we may write that a sequence of vertices $v_0, \ldots, v_k$ is a \emph{walk} of $D$ to express that some sequence of edges $v_0v_1, \ldots, v_{k-1}v_k$ is a walk of $D$.  Both conventions are ambiguous, as there could be $q$ edges from $v_i$ to $v_{i+1}$ (or from $v_k$ to $v_0$ in the former case).  In general, the edge represented by $v_iv_{i+1}$ is clear by context; if not, then $v_iv_{i+1}$ refers to any of the edges from $v_i$ to $v_{i+1}$.  

An \emph{edge weighing}, or simply a \emph{weighing}, of a $q$-digraph $D$ is a function $w\colon E(D) \to \mathbb{R}$.  The value $w(uv)$ is referred to as the \emph{weight} of $uv$ (with respect to $w$).  For any bag of edges $E$, the \emph{weight} of $E$ (with respect to an edge weighing $w$) is $w(E) = \sum_{uv \in E}w(uv)$.  

Recall that a PCA model $\M = (C,\A)$ is a $(\Circ,\Len)$-CA model when: 1.\ $|C| = \Circ$; 2.\ all the arcs in $\A$ have length $\Len$; and 3.\ $|e,e'| \geq \Dist = 1$ for every pair of consecutive extremes $e$ and $e'$.  Clearly, if we let $\BegDist = 0$, then 4.\ $|s(A), s(A')| \geq \Dist + \BegDist = 1$ for ever pair of arcs $A$ and $A'$.  Although our arbitrary choices for $\Dist$ and $\BegDist$ may seem natural, in some applications it is better to allow $\Dist$ and $\BegDist$ to take different values~\cite{KlavikKratochvilOtachiRutterSaitohSaumellVyskocilA2017,SoulignacJGAA2017a}.  For this reason, we say that a tuple $\Desc = (\Circ, \Len, \Dist, \BegDist)$ is a \emph{UCA descriptor} when $\Circ,\Len,\Dist \in \mathbb{R}_{>0}$ and $\BegDist \in\mathbb{R}_{\geq0}$, while $\M$ is a $\Desc$-CA model when it satisfies 1--4 for the values in $\Desc$.  For the sake of notation, we may also say use a pair $(\Circ, \Len)$ in place of a UCA descriptor; in such cases, $\Dist = 1$ and $\BegDist=0$.

Our new terminology allows for a better description of what a minimal model is.  For a UCA descriptor $\Desc = (\Circ, \Len, \Dist, \BegDist)$, say that a $\Desc$-CA model $\Unit$ is \emph{$(\Dist, \BegDist)$-minimal} (resp.\ \emph{$(\Dist, \BegDist)$-minimum}) when $\Len \leq \Len'$ and $\Circ \leq \Circ'$ for every $(\Circ', \Len', \Dist, \BegDist)$-CA model equivalent (resp.\ isomorphic) to $\Unit$.  We omit the parameters for the special case in which $\Dist=1$ and $\BegDist =0$.  The following non-trivial theorems will be taken for granted in the rest of the article.

\begin{theorem}[\cite{PirlotTaD1990,SoulignacJGAA2017}]
  Every UCA (resp.\ UIG) model is equivalent to some $(\Dist, \BegDist)$-minimal UCA model, for all $\Dist \in \mathbb{R}_{>0}$ and $\BegDist \in \mathbb{R}_{\geq0}$.
\end{theorem}

\begin{theorem}[\cite{PirlotTaD1990}]
  If a $(\Circ, \Len, \Dist, \BegDist)$-CA model $\Unit$ with no external arcs is $(\Dist, \BegDist)$-minimal, then $\Circ$ and $\Len$ are integer combinations of $\Dist$ and $\BegDist$.
\end{theorem}

\section{Synthetic Graphs}
\label{sec:synthetic graph}

This section introduces synthetic graphs, their associated weighing $\Sep$, and Mitas' drawings. The presentation summarizes the features that we require in this work; for motivations and a thorough explanation of its inception, we refer to~\cite{Mitas1994,PirlotVincke1997,PirlotTaD1990,SoulignacJGAA2017,SoulignacJGAA2017a}.

Let $\M = (C, \A)$ be a PCA model with arcs $A_0 < \ldots < A_{n-1}$.  The \emph{synthetic graph} of $\M$ is the $3$-digraph $\Syn(\M)$ (\Figure~\ref{fig:synthetic graph}(b)) that has a vertex $v(A_i)$ for each $A_i \in \A$ and whose bag of edges is $E_\Steps \cup E_\Noses \cup E_\Hollows$, where:
\begin{itemize}
 \item $E_\Steps = \{v(A_i) \to v(A_{i+1}) \mid i \in \Range{n}\} \cup \{v(A_{n-1}) \to v(A_0)\}$,
 \item $E_\Noses = \{v(A_i) \to v(A_j) \mid t(A_i)s(A_j) \mbox{ are consecutive in } \M\}$, and
 \item $E_\Hollows = \{v(A_i) \to v(A_j) \mid s(A_i)t(A_j) \mbox{ are consecutive in } \M\}$.
\end{itemize}
The edges in $E_\Steps$, $E_\Noses$, and $E_\Hollows$ are the \emph{steps}, \emph{noses}, and \emph{hollows} of $\Syn(\M)$, respectively.  We drop the parameter $\M$ from $\Syn$ when no ambiguities are possible, and we implicitly consider the definitions on $\Syn$ as being valid on $\M$, and vice versa, when no confusions are possible.  Moreover, we regard the arcs of $\M$ as being the vertices of $\Syn$, thus we say that $A_i \to A_j$ is a nose instead of writing that $v(A_i) \to v(A_j)$ is a nose.  

\begin{figure}[t!]%
  \centering
  \begin{tabular}{c@{\hspace{1cm}}c@{\hspace{1cm}}c}
    \includegraphics{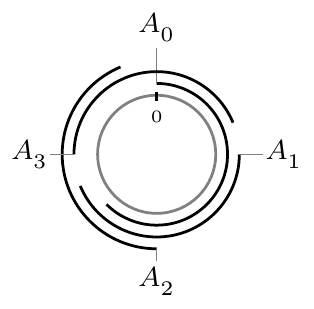} & \includegraphics{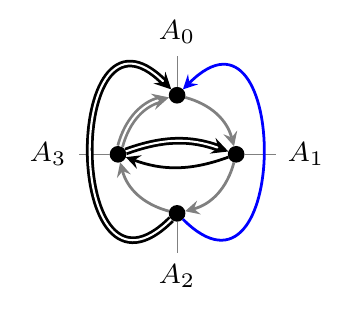} & \includegraphics{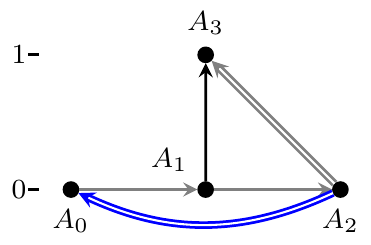} \\
    (a) & (b) & (c)
  \end{tabular}
  \caption{(a) A PCA model $\M$ with arcs $A_0 < A_1 <A_2<A_3$, (b) its synthetic graph $\Syn$, and (c) Mitas' drawing of $\Syn$ with backward edges and row numbers.  Black, blue, and gray lines represent noses, hollows, and steps, respectively, while double lines represent external (b) and backward (c) edges.}\label{fig:synthetic graph}
\end{figure}

The edges of $\Syn$ are classified into internal or external according to the way they interact with $0$ (\Figure~\ref{fig:synthetic graph}(b)).  A step (resp.\ nose) $A_i \to A_j$ is \emph{internal} (resp.\ \emph{external}) if and only if $(s(A_i), s(A_j))$ is internal (resp.\ external), while a hollow $A_i \to A_j$ is \emph{internal} (resp.\ \emph{external}) if and only if $(s(A_j), s(A_i))$ is internal (resp.\ external).  

Each UCA descriptor $\Desc = (\Circ, \Len, \Dist, \BegDist)$ implies a weighing $\Sep_{\Desc}$ of the edges of $\Syn$ whose purpose is to indicate how far or close $s(A_i)$ and $s(A_j)$ must be in any $\Desc$-CA model equivalent to $\M$. For each edge $A_i \to A_j$ of $\Syn$, let $q_{ij} \in \{0,1\}$ be equal to $0$ if and only if $A_i \to A_j$ is internal, and define
\begin{displaymath}
 \Sep_{\Desc}(A_i \to A_j) =
  \begin{cases}
   \Dist + \BegDist - \Circ q_{ij} & \mbox{if $A_i \to A_j$ is a step} \\
   \Dist + \Len - \Circ q_{ij} & \mbox{if $A_i \to A_j$ is a nose, and} \\
   \Dist - \Len + \Circ q_{ij} & \mbox{if $A_i \to A_j$ is a hollow.}
  \end{cases} 
\end{displaymath}

Let $\Noses(\W)$, $\Hollows(\W)$, and $\Steps(\W)$ (resp.\ $\Noses_{\Ext}(\W)$, $\Hollows_{\Ext}(\W)$, and $\Steps_{\Ext}(\W)$) be the number of (resp.\ external) noses, hollows, and steps of a walk $\W$, respectively.  Recall that, viewing $\W$ as a bag of arcs, its weight is $\Sep_{\Desc}(\W) = \sum_{A_i \to A_j \in \W}\Sep_{\Desc}(A_i \to A_j)$.  Then, 
\begin{align}
  \Sep_{\Desc}(\W) &= \Len\Jump(\W) + \Circ\Ext(\W) + \Dist|\W| + \BegDist\Steps(\W) \mbox{,} \label{eq:sep}
\end{align}
where $\Jump(\W) = \Noses(\W) - \Hollows(\W)$, and $\Ext(\W) = \Hollows_{\Ext}(\W) - \Noses_{\Ext}(\W) - \Steps_{\Ext}(\W)$.

\begin{theorem}[\cite{PirlotTaD1990,SoulignacJGAA2017,SoulignacJGAA2017a}]\label{thm:separation constraints}
 A PCA model\/ $\M$ is equivalent to a $\Desc$-CA model if and only if $\Sep_\Desc(\W) \leq 0$ for every cycle $\W$ of $\Syn$. Furthermore, there exists an $O(n^2)$ time and $O(n)$ space algorithm that, given a UCA descriptor $\Desc$, outputs either a $\Desc$-CA model equivalent to $\M$ or a cycle $\W$ of $\Syn$ with $\Sep_\Desc(\W) > 0$.
\end{theorem}

The synthetic graph $\Syn$ admits a peculiar drawing in which its vertices occupy entries of an imaginary matrix.  For $i \in \Range{n}$, let $\Ind(A_i)$ be the maximum number of pairwise non-intersecting arcs in $\{A_0, \ldots, A_i\}$.  The \emph{row} of $A_i$ is $\Row(A_i) = \Ind(A_i) - 1$, while the \emph{number of rows} of $\M$ is $\Ind(A_{n-1})$.  The maximal sequence $A_j < \ldots < A_k$ of arcs with row $r$, for $r \in \Range{\Rows(\M)}$, is the \emph{row $r$} of $\M$, while $A_j$ and $A_k$ are the \emph{leftmost} and \emph{rightmost} at row $r$ (\Figure~\ref{fig:synthetic graph}(c)).

Say that a step (resp.\ nose, hollow) $A_i \to A_j$ of $\Syn$ is a $\delta$-step (resp.\ $\delta$-nose, $\delta$-hollow) when $\Row(A_j) - \Row(A_i) = \delta$.  We refer to $0$-steps, $1$-noses, and $(-1)$-hollows as being \emph{forward} edges, and to $1$-steps and $0$-hollows as being \emph{backward} edges.  It is not hard to see that an edge is internal if and only if it is either forward or backward.  We say a walk $\W$ of $\Syn$ is \emph{internal} when it contains only internal edges, and that is \emph{forward} when it contains only forward edges.

A key observation by Mitas~\cite{Mitas1994} is that the digraph $\SynInt$ obtained after removing the external and backward edges of $\Syn$ is acyclic.  This fact allows us to define the column of the vertices in $\Syn$ using the following recurrence.  The \emph{column} $A_0$ is $\Col(A_0) = 0$, while, for $0 < \varepsilon \ll 1/n$ and $i \in \Range{n}$, the \emph{column} of $A_i$ is:
\begin{align}
  \Col(A_i) = \max\left\{
    \begin{array}{r}
       \Col(N) + \varepsilon \\ \Col(H) + 1 \\ \Col(S) + 1
    \end{array} \ \middle| \ 
    \begin{array}{l}
      N \to A_i \mbox{ is a $1$-nose } \\
      H \to A_i \mbox{ is a $-1$-hollow } \\
      S \to A_i \mbox{ is a $0$-step } \\
    \end{array}\right\}\label{eq:column}
\end{align}

\emph{Mitas' drawing} (\Figure~\ref{fig:synthetic graph}(c)) is obtained by placing each vertex $A_i$ in the plane at point $\Pos(A_i) = (\Row(A_i), \Col(A_i))$, and joining $\Pos(A_i)$ with $\Pos(A_j)$ with a straight line $L(A_iA_j)$, for each edge $A_i \to A_j$ of $\SynInt$. Clearly, any forward walk $\W = B_0, \ldots, B_{k-1}$ of $\Syn$ is also a walk of $\SynInt$; let $\Gr(\W)$ be the curve obtained by traversing $L(B_{i}B_{i+1})$ after $L(B_{i-1}B_{i})$, for $i \in \Range{k-1}$.

\begin{observation}[\cite{Mitas1994,SoulignacJGAA2017a}]\label{obs:walk drawing}
  If\/ $\W$ is a forward walk of\/ $\Syn$, then $\Gr(\W)$ is the graph of a continuous function in $\mathbb{R} \to \mathbb{R}$.
\end{observation}

Removing the backward edges of an internal walk $\W$, we obtain a family of forward walks $\W_0, \ldots, \W_{k-1}$.  The \emph{drawing} of $\W$ is $\Gr(\W) = \bigcup_{i \in \Range{k}} \Gr(\W_i)$.  Mitas' drawing of $\SynInt$ is so attractive because it is a ``plane'' drawing~\cite{Mitas1994,SoulignacJGAA2017a}.  

\begin{theorem}[\cite{Mitas1994,SoulignacJGAA2017a}]\label{thm:plane drawing}
  Two internal walks $\W$ and\/ $\W'$ of $\Syn$ have a common vertex if and only if $\Gr(\W) \cap \Gr(\W') \neq \emptyset$. Furthermore, $A$ is a vertex common to $\W$ and\/ $\W'$ if and only if\/ $\Pos(A) \in \Gr(\W) \cap \Gr(\W')$.
\end{theorem}

\begin{figure}[t!]
  \centering
  \begin{tabular}{ccc}
    \includegraphics{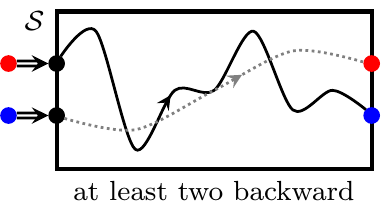} & \includegraphics{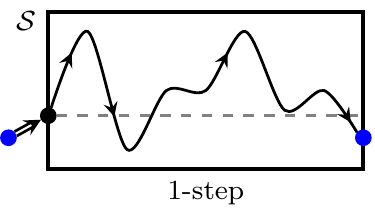} & \includegraphics{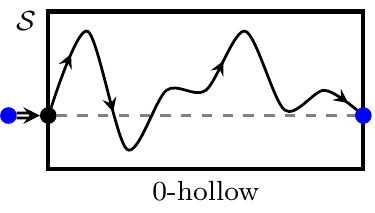} \\
    (a) & (b) & (c)
  \end{tabular}
  \caption{Any backward edge ($0$-hollow or $1$-step) that begins at row $r$ ends at row either $r$ or $r+1$.  Then, by Observation~\ref{obs:walk drawing} and Theorem~\ref{thm:plane drawing}, an internal circuit is a cycle if and only if it contains exactly one backward edge (a).  Moreover, $1$-noses (resp.\ $0$-steps, $(-1)$-hollows) that begin at $r$ end at $r+1$ (resp.\ $r$, $r-1$), thus any internal cycle contains exactly one more hollow than noses (b)~and~(c), i.e., it has $\Jump = -1$.}\label{fig:one more hollow}
\end{figure}

To highlight the utility of Mitas' drawings, Figure~\ref{fig:one more hollow} contains an informal geometrical proof of the next corollary that exploits Observation~\ref{obs:walk drawing} and Theorem~\ref{thm:plane drawing}.  

\begin{corollary}[\cite{PirlotTaD1990,SoulignacJGAA2017a}]\label{cor:one more hollow}
  If\/ $\W$ is an internal cycle of\/ $\Syn$, then $\Jump(\W) = -1$.
\end{corollary}

\section{A New Characterization of UCA Models}
\label{sec:characterization}

\begin{figure}[b!]
  \centering
  \begin{tabular}{c@{\hspace{1cm}}c}
    \includegraphics{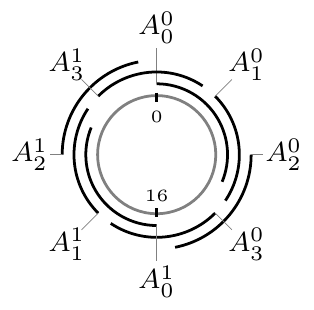} & \includegraphics{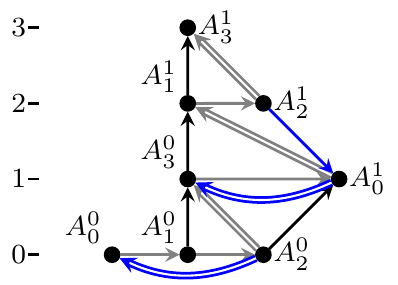} \\
    (a) & (b)
  \end{tabular}
  \caption{(a) $2 \Unroll \M$ for the model $\M$ in \Figure~\ref{fig:synthetic graph}, where $A_i^j$ is the $j$-th copy of $A_i$. (b) Mitas' drawing of $2 \Unroll \Syn$ with backward edges.}\label{fig:unrolling}
\end{figure}

In this section we introduce the loop unrolling technique to prove a new characterization of those PCA models that have equivalent UCA models.  

Let $\Circ$ be the circumference of the circle of a PCA model $\M$.  The \emph{$\Copies$-unrolling} of $\M$ (\Figure~\ref{fig:unrolling}) is the PCA model $\Copies \Unroll \M$ whose circle has circumference $\Copies\Circ$ that has $\Copies$ arcs $A_0, A_1, \ldots, A_{\Copies-1}$ for every $A \in \A$ such that, for $i \in \Range{\Copies}$:
\begin{displaymath}
  s(A_i)=s(A) + i\Circ\mbox{, and } t(A_i)=t(A) + \Circ(i+q) \bmod \Copies\Circ,
\end{displaymath}
where $q \in \{0,1\}$ equals $1$ if and only if $A$ is external.
For convenience, we write $\Copies \Unroll \Syn(\M)$ as a shortcut for $\Syn(\Copies \Unroll \M)$, and we drop the parameter $\M$ when no confusions are possible.  By definition, the arc $A_i$ of $\Copies \Unroll \M$ is a vertex of $\Copies \Unroll \Syn$ for every $i \in \Range{\Copies}$ (\Figure~\ref{fig:unrolling}(b)).  We refer to $A_i$ as being the \emph{$i$-th copy} of both $A$ and $A_j$ (for $j \in \Range{\Copies}$).  Similarly, each edge $A_i \to B_j$ of $\Copies \Unroll \Syn$ is said to be a \emph{copy} of the edge $A \to B$ of $\Syn$ to indicate that $B_j$ is a copy of $B$ and $A_i \to B_j$ is of the same kind as $A \to B$, while each walk $\T$ of $\Copies\Unroll\Syn$ is a \emph{copy} of the walk $\W$ of $\Syn$ such that the $i$-th edge of $\T$ is a copy of the $i$-th edge of $\W$, for $i \in \Range{\T}$.  Note that $A \to B$ has $\Copies$ copies in $\Copies \Unroll \Syn$ by definition.  We remark, as it can be observed in \Figure~\ref{fig:unrolling}, that $\Row(B_j) - \Row(A_i)$ need not be equal to $\Row(B) - \Row(A)$. Thus, the $\delta$-noses (resp.\ $\delta$-hollows, $\delta$-steps) of $\Copies \Unroll \Syn$ need not correspond to the $\delta$-noses (resp.\ $\delta$-hollows, $\delta$-steps) of $\Syn$.

Our characterization of those PCA models that admit equivalent UCA models is given below.  Nose and hollow walks have a central role in our theorem, as they do in Tucker's characterization~\cite{SoulignacJGAA2017,TuckerDM1974}.  A walk $\W = B_0, \ldots, B_{k-1}$ of $\Syn$ is a \emph{nose walk} (resp.\ \emph{hollow walk}) when it contains no hollows (resp.\ noses).  Walk $\W$ is \emph{greedy} when either $B_i \to B_{i+1}$ is a nose (resp.\ hollow) or there is no nose (resp.\ hollow) from $B_i$ in $\Syn$, for every $i \in \Range{k}$. That is, $\W$ is greedy when noses (resp.\ hollows) are preferred over steps. 

\begin{theorem}\label{thm:crossing cycles}
  The following statements are equivalent for a PCA model $\M$.%
  \begin{enumerate}
    \item $\M$ is equivalent to a UCA model.\label{thm:crossing cycles:uca}
    \item Every pair of circuits of\/ $\Syn$ with different signs of $\Ext$ have a common vertex.\label{thm:crossing cycles:every ext}
    \item Some greedy hollow cycle and some greedy nose cycle of\/ $\Syn$ share a vertex.\label{thm:crossing cycles:some greedy}
  \end{enumerate}
\end{theorem}

\begin{proof}
  \ref{thm:crossing cycles:uca} $\Rightarrow$ \ref{thm:crossing cycles:every ext}.  First consider the case in which $\W_N$ and $\W_H$ are circuits of $\Syn$ with $\Ext(\W_N) < 0 < \Ext(\W_H)$. This means that, for any $\Copies \geq 1$, every row $r < \Copies-1$ of $\Copies \Unroll \Syn$ has at least one copy of a vertex in $\W_N$ and one copy of a vertex in $\W_H$.  Define $w=\max\{n,|\W_N|,|\W_H|\}$, and take $\Copies \gg w$ to be large enough.  Let $N_w$ and $H_w$ be copies of vertices in $\W_N$ and $\W_H$ that belong to row $w$, respectively.  By traversing $i$ copies of $\W_N$ from $N_w$, we obtain a walk $\T_N(i)$ that ends at some copy $N_{x(i)}$ of $N_w$ whose row is $x(i) > w$.  Similarly, if we traverse $i$ copies of $\W_H$ in reverse from $H_w$, we obtain a walk $\T_H(i)$ that begins at some copy $H_{y(i)}$ of $H_w$ whose row is $y(i) > w$.  By definition, each row $r < \Copies-1$ of $\Copies\Unroll\Syn$ is uniquely determined by its leftmost vertex, thus there exist $a,a',b,b' \in \Range{n^3}$ such that: $x = x(a) = y(a')$, $y = x(b) = y(b')$ and row $y$ is a copy of row $x$.  Furthermore, as $|\W_N| \leq w$ and $|\W_H| \leq w$ and $\Copies$ is large enough (say $\Copies \gg w(n^3+1)$), the walks $\T_H$ and $\T_N$ that join $H_x$ and $N_y$ to $H_y$ and $N_x$, respectively, are internal (\Figure~\ref{fig:crossing cycles}).  
  
  \begin{figure}[t!]
    \centering
    \begin{tabular}{c@{\hspace{3mm}}c@{\hspace{3mm}}c}
    \includegraphics{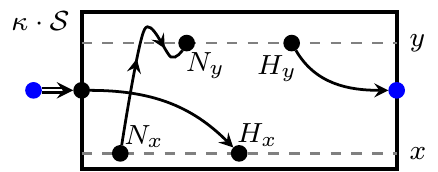} & \includegraphics{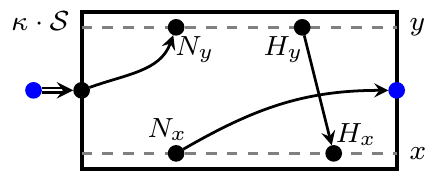} & \includegraphics{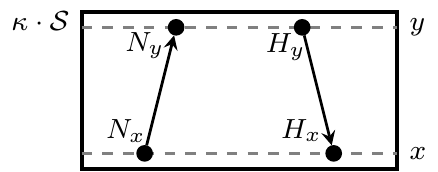}\\
    (a) & (b) & (c)
    \end{tabular}
    \caption{Proof of \ref{thm:crossing cycles:uca} $\Rightarrow$ \ref{thm:crossing cycles:every ext} in Theorem~\ref{thm:crossing cycles}.}\label{fig:crossing cycles}
  \end{figure}

  By Observation~\ref{obs:walk drawing} and Theorem~\ref{thm:plane drawing}, $\T_N$ and $\T_H$ have a common vertex when at least one of them is not forward (\Figure~\ref{fig:crossing cycles}(a)(b)), and so do $\W_N$ and $\W_H$.  Suppose, then, that both $\T_N$ and $\T_H$ are forward walks (\Figure~\ref{fig:crossing cycles}(c)).  By construction, $\T_N$ and $\T_H$ are copies of the circuits $\W_N'$ and $\W_H'$ of $\Syn$ that are obtained traversing $(b-a)$ and $(b'-a')$ times $\W_N$ and $\W_H$, respectively.  Moreover, $\W_N'$ has the same number of external edges as $\W_H'$ because $\T_N$ and $\T_H$ both join rows $x$ and $y$.  Then, for any UCA descriptor $\Desc = (\Circ, \Len)$, we obtain that 
  \begin{align*}
    \Sep_{\Desc}(\W_H') &= \Sep_{\Desc}(\T_H) + c\Ext(\W_H') \geq (y-x)(1-\Len) - c\Ext(\W_N')\mbox{, and}\\
    \Sep_{\Desc}(\W_N') &= \Sep_{\Desc}(\T_N) + c\Ext(\W_N') \geq (y-x)(\Len+1) + c\Ext(\W_N'), 
  \end{align*}
  thus $\Sep_{\Desc}(\W_H') + \Sep_{\Desc}(\W_N') > 0$, and $\M$ is not equivalent to a UCA model by Theorem~\ref{thm:separation constraints}.
  
  Now consider the case in which $\Ext(\W) = 0$ and $\Ext(\W') \neq 0$ for $\{\W, \W'\} = \{\W_N, \W_H\}$.  As before, we can assure that some internal copy $\T'$ of $\W'$ in $(4n)\Unroll\Syn$ joins a vertex at row $n$ with a vertex at row $3n$.   On the other hand, some internal copy $\T$ of $\W$ in $(4n)\Unroll\Syn$ has a backward edge joining a rightmost vertex at row $i$ to a leftmost vertex at row $j$ for $i, j \in \Range{2n,3n}$ (recall $|i-j| \leq 1$).  By Theorem~\ref{thm:plane drawing} and Observation~\ref{obs:walk drawing}, $\T$ and $\T'$ have a common vertex, and so do $\W_N$ and $\W_H$.
  
  \ref{thm:crossing cycles:every ext} $\Rightarrow$ \ref{thm:crossing cycles:uca}.  Let $\Len = 2n^2$ and $\Circ$ be the minimum such that $\Sep_{(\Circ, \Len)} \leq 0$ for every cycle $\W$ of $\Syn$ with $\Ext < 0$.  Note that such a value of $\Circ$ always exists by~\eqref{eq:sep}.  Moreover, some cycle $\W_N$ of $\Syn$ with $\Ext < 0$ has $\Sep_{(\Circ, \Len)} = 0$.  We prove that $\Sep_{(\Circ, \Len)}(\W) \leq 0$ for every cycle $\W$ of $\Syn$, thus $\M$ is equivalent to a $(\Circ,\Len)$-CA model by Theorem~\ref{thm:separation constraints}.  
  \begin{description}
    \item [Case 1:] $\Ext(\W) < 0$.  Then $\Sep_{(\Circ, \Len)}(\W) \leq 0$ by the definition of $\Circ$.
    \item [Case 2:] $\Ext(\W) = 0$.  This follows by~\eqref{eq:sep} and Corollary~\ref{cor:one more hollow} because $|\W| \leq n$.
    \item [Case 3:] $\Ext(\W) > 0$.  By hypothesis, $\W$ and $\W_N$ have a common vertex $A$.  Let $\W_0$ be the circuit of $\Syn$ that begins at $A$ which is obtained by traversing $|\Ext(\W_N)|$ times $\W$ and then $\Ext(\W)$ times $\W_N$.  Clearly, $|\W_0| \leq 2n^2 = \Len$ and $\Ext(\W_0) = 0$.  The latter implies that $\W_0$ has some internal copy $\T_0$ in $\Copies \Unroll \Syn$ when $\Copies$ is large enough (\Figure~\ref{fig:repeated cycle}).  Clearly, $\Jump(\W_0) = \Jump(\T_0)$, thus $\Jump(\W_0) < 0$ by Corollary~\ref{cor:one more hollow}.  Then, by~\eqref{eq:sep},
    \begin{gather*}
      |\Ext(\W_N)|\Sep_{(\Circ, \Len)}(\W) = |\Ext(\W_N)|\Sep_{(\Circ, \Len)}(\W) + \Ext(\W)\Sep_{(\Circ, \Len)}(\W_N) =\\ 
                           \Sep_{(\Circ, \Len)}(\W_0) = \Len\Jump(\W_0) + |\W_0| \leq -\Len + |\W_0| \leq 0
    \end{gather*}
  \end{description}

  \begin{figure}[t!]
    \mbox{}\hfill \includegraphics{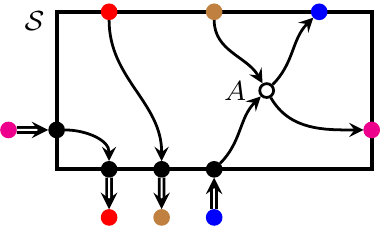} \hfill \includegraphics{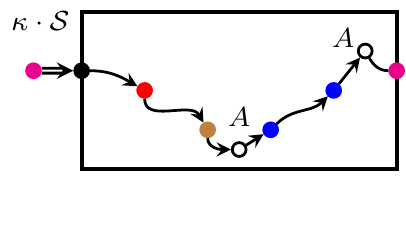} \hfill\mbox{}
    \caption{Theorem~\ref{thm:crossing cycles} (\ref{thm:crossing cycles:every ext} $\Rightarrow$ \ref{thm:crossing cycles:uca}): circuit $\W_0$ and its copy $\T_0$ in $\Copies\Unroll\Syn$.}\label{fig:repeated cycle}
  \end{figure}

  \ref{thm:crossing cycles:every ext} $\Rightarrow$ \ref{thm:crossing cycles:some greedy} is trivial.
  
  \ref{thm:crossing cycles:some greedy} $\Rightarrow$ \ref{thm:crossing cycles:every ext}. Suppose some greedy nose cycle $\W_N$ has a vertex $A$ in common with a greedy hollow cycle $\W_H$.  Fix a large enough $\Copies$ and let $\T_N$ be a walk of $\Copies\Unroll\Syn$ obtained by traversing $3\Ext(\W_H)$ times $\W_N$ from some copy $A_0$ of $A$ in $\Copies \Unroll \Syn$, ending at some other copy $A_3$ of $A$. Similarly, let $\T_H$ be a walk of $\Copies\Unroll\Syn$ obtained by traversing $3\Ext(\W_N)$ times $\W_H$ from $A_3$, ending at $A_0$. It is easily seen that $\Copies$ and $A_0$ can be chosen so that $n < \Row(A_0) < \Row(A_3) < \Copies-n$, which implies that $\T_N$ and $\T_H$ are internal in $\Copies \Unroll \Syn$. Notice that besides $A_0$ and $A_3$, $\T_N$ and $\T_H$ have copies $A_1$ and $A_2$ of $A$ in common. Also, observe that $\T_N$ is forward because $\W_N$ is greedy, so, by Corollary~\ref{cor:one more hollow}, the subpath of $\T_H$ from $A_{i+1}$ to $A_i$ is not forward for $i\in\Range{3}$ (\Figure~\ref{fig:greedy crossings}).  Hence, $\T_H$ contains at least two subpaths $\T_0$ and $\T_1$ that join a leftmost vertex to a rightmost vertex. Let $x_i$ and $y_i$ be the rows of the leftmost and rightmost vertices of $\T_i$ for $i \in \{0,1\}$ (\Figure~\ref{fig:greedy crossings}).

  \begin{figure}[t!]
    \centering \includegraphics{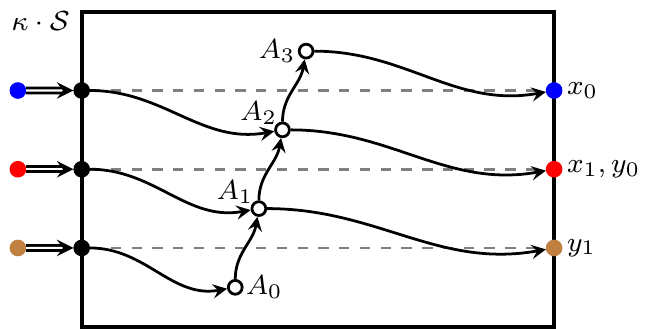}
    \caption{Theorem~\ref{thm:crossing cycles} (\ref{thm:crossing cycles:some greedy} $\Rightarrow$ \ref{thm:crossing cycles:every ext}): walks $\T_N$ and $\T_H$ in $\Copies\Unroll\Syn$.}\label{fig:greedy crossings}
  \end{figure}

  By repeatedly traversing copies of any circuit $\W$ of $\Syn$ with $\Ext(\W) > 0$, we obtain a walk $\T$ in $\Copies\Unroll\Syn$ that joins a vertex at row $x_i+1$ with a vertex at row $y_i-1$.  If $\T$ has a vertex $B$ in common with $\T_i$ and the edge $B \to B'$ from $B$ in $\T$ is a hollow, then $B \to B'$ must also an edge of $\T_i$ because $\T_i$ is greedy. By Observation~\ref{obs:walk drawing}, this implies that $\T$ must pass through a backward edge before reaching the level $y_i-1$. Consequently, $\T$ contains at least two backward edges and, so, by Observation~\ref{obs:walk drawing} and Theorem~\ref{thm:plane drawing}, $\W$ has a vertex in common to every circuit $\W'$ with $\Ext(\W') \leq 0$.  Similarly, by Observation~\ref{obs:walk drawing} and Theorem~\ref{thm:plane drawing}, any circuit with $\Ext < 0$ has a vertex in common with every circuit with $\Ext = 0$, as proven in \ref{thm:crossing cycles:uca} $\Rightarrow$ \ref{thm:crossing cycles:every ext}.
\end{proof}

\paragraph{A comparison to Tucker's characterization.} The first characterization of those PCA models that have equivalent UCA models was given by Tucker~\cite{TuckerDM1974}.  To translate Tucker's characterization to the language of synthetic graphs, Soulignac~\cite{SoulignacJGAA2017} defines the so-called ``nose ratio'' $r$ and ``hollow ratio'' $R$.  We shall not recall these definitions here, because they are rather technical and non-important for our article.  Yet, we recall Tucker's theorem, as stated by Soulignac, for the sake of the comparison.

\begin{theorem}[\cite{TuckerDM1974,SoulignacJGAA2017}]\label{thm:tucker+soulignac}
  A PCA model $\M$ has equivalent UCA models if and only $r(\W_N) < R(\W_H)$ for every greedy nose cycle $\W_N$ and every greedy hollow cycle $\W_H$ of $\Syn$.
\end{theorem}

Tucker's theorem is the basis for the three polynomial time algorithms that output a negative witness certifying that $\M \not\in\mathbb{M}$~\cite{DuranGravanoMcConnellSpinradTuckerJA2006,KaplanNussbaumDAM2009,SoulignacJGAA2017a}.  In a nutshell, these algorithms compute all the greedy nose and greedy hollow cycles, and then they compare their ratios. Yet, by Theorem~\ref{thm:crossing cycles}, only one greedy nose and one greedy hollow cycle need to be computed.  Furthermore, there is no need to compute the ratios; we only have to make sure that both cycles have a common vertex. Clearly, the algorithm so obtained can be implemented to run in linear time.  Regarding the time complexity, we remark that, although there is no improvement over~\cite{KaplanNussbaumDAM2009,SoulignacJGAA2017a} in the worst case, it provides a faster algorithm when the greedy cycles of $\Syn$ are short and the data structure representing the input model $\M$ allows the efficient computation of noses and hollows.  This is a common case when the $\M$ is obtained by running a recognition algorithm on a PCA graph~\cite{DengHellHuangSJC1996,SoulignacA2015}.  More important than this is the fact that Theorem~\ref{thm:crossing cycles}, when combined with Mitas' drawings of the unrolled synthetic graph, allows us to better visualize the structure of UCA models.

\section{The Integrality of \texorpdfstring{$\Circ$}{\it c} and \texorpdfstring{$\Len$}{\it l}}
\label{sec:integrality}

The purpose of this section is to prove that both $\Circ$ and $\Len$ are integer combinations of $\Dist$ and $\BegDist$ when $\Unit$ is a $(\Dist, \BegDist)$-minimal $\Desc$-CA model.  Pursuing our goal, we first show that $\Unit$ has some special circuits with $\Sep_\Desc = 0$.  These circuits are later combined with~\eqref{eq:sep} to prove our main result.

\begin{lemma}\label{lemma:minimal conditions}
 If\/ $\Unit$ be a $(\Dist,\BegDist)$-minimal $\Desc$-CA model for $\Desc=(\Circ, \Len, \Dist, \BegDist)$, then:
 \begin{enumerate}[(a)]
  \item $\Syn$ has cycles $\W_N$ and $\W_H$ with $\Sep_{\Desc}(\W_N) = \Sep_{\Desc}(\W_H) = 0$ such that $\Ext(\W_N) < 0$ and $\Ext(\W_H) \geq 0$.\label{lemma:minimal conditions:ext for sep=0}
  \item $\Syn$ has a circuit $\W_0$ with $\Sep_{\Desc}(\W_0) = \Ext(\W_0) = 0$.\label{lemma:minimal conditions:ext=0}
  \item $\Syn$ has a circuit $\W_1$ with $\Sep_{\Desc}(\W_1) = 0$ and $\Ext(\W_1) = -1$.\label{lemma:minimal conditions:ext=-1}
 \end{enumerate}
\end{lemma}

\begin{proof}
  (\ref{lemma:minimal conditions:ext for sep=0}) Let
 \begin{displaymath}
  \Delta = \min\{-\Sep_\Desc(\W) \mid \W \text{ is a cycle of $\Syn$ with } \Sep_\Desc(\W) \neq 0\},
 \end{displaymath}
 which is positive by Theorem~\ref{thm:separation constraints}.  Consider the following cases.
 \begin{description}
   \item [Case 1:] $\Ext(\W) \geq 0$ for every cycle with $\Sep_\Desc(\W) = 0$.  Let $v = (\Circ - \frac{\Delta}{2n}, \Len, \Dist, \BegDist)$ be a UCA descriptor.  By~\eqref{eq:sep} and the fact that $|\Ext(\W)| \leq n$ is an integer value for every cycle $\W$ of $\Syn$, we obtain that
   \begin{align*}
     \Sep_{v}(\W) = \Sep_\Desc(\W) - \frac{\Delta\Ext(\W)}{2n} \leq \begin{cases}
       0   & \text{if } \Sep_\Desc(\W) = 0 \\
       -\Delta + \frac{\Delta}{2} = -\frac{\Delta}{2} & \text{otherwise.}
     \end{cases}
   \end{align*}
   Therefore, $\Unit$ is equivalent to a $v$-CA by Theorem~\ref{thm:separation constraints}, implying that $\Unit$ is not $(\Dist, \BegDist)$-minimal.
   
   \item [Case 2:] $\Ext(\W) < 0$ for every cycle with $\Sep_\Desc(\W) = 0$.  Let $v = (\Circ + \frac{\Delta}{2n}, \Len - \frac{\Delta}{4n^2}, \Dist, \BegDist)$ be a UCA descriptor.  By~\eqref{eq:sep} and the facts that $|\Ext(\W)| \leq n$ is integer and $|\Jump(\W)| \leq n$ for every cycle $\W$ of $\Syn$, we obtain that 
   \begin{align*}
   \Sep_{v} = \Sep_\Desc - \frac{\Delta\Jump}{4n^2} + \frac{\Delta\Ext}{2n} \leq 
   \begin{cases}
     0 + \frac{\Delta}{4n} - \frac{\Delta}{2n} = -\frac{\Delta}{4n} & \text{if} \Sep_\Desc = 0\\
     -\Delta + \frac{\Delta}{4n} + \frac{\Delta}{2} \leq -\frac{\Delta}{4n} & \text{otherwise,}
   \end{cases}
   \end{align*}
   where $\W$ is the omitted parameter. As in Case~1, $\Unit$ is not $(\Dist, \BegDist)$-minimal.\hfill$\triangle$
 \end{description}
  
 (\ref{lemma:minimal conditions:ext=0}) Let $\W_N$ and $\W_H$ be the cycles implied by~(\ref{lemma:minimal conditions:ext for sep=0}), and $A$ be an arc of both $\W_N$ and $\W_H$, that exists by Theorem~\ref{thm:crossing cycles}.  Clearly, the circuit $\W_0$ obtained by traversing $|\Ext(\W_N)|$ times $\W_H$ plus $\Ext(\W_H)$ times $\W_N$, starting from $A$, has $\Sep_{\Desc} = \Ext = 0$.\hfill$\triangle$
 
 \begin{figure}
  \centering
  \begin{tabular}{c@{\hspace{1cm}}c}
    \includegraphics{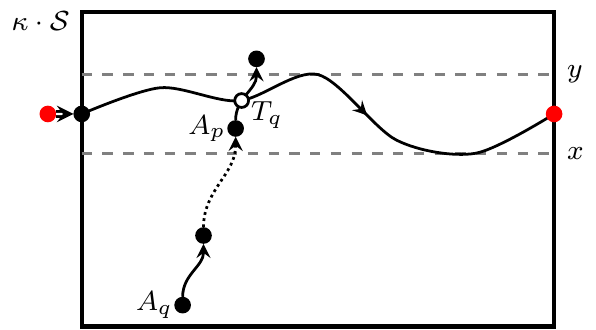} & \includegraphics{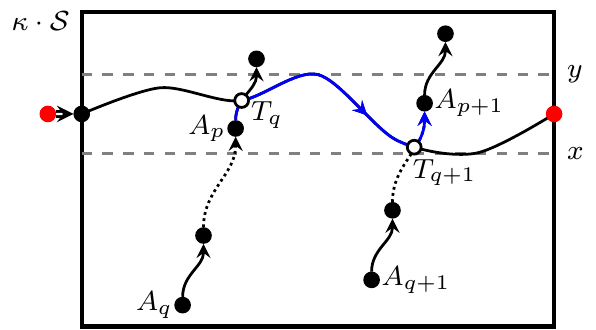}\\
    (a)  &  (b)
  \end{tabular}
  \caption{Proof of Lemma~\ref{lemma:minimal conditions} (\ref{lemma:minimal conditions:ext=-1}); here $p = ae+q$.}\label{fig:minimal conditions:ext=-1}
 \end{figure}

 (\ref{lemma:minimal conditions:ext=-1})  Let $\W_N$ be the cycle implied by (\ref{lemma:minimal conditions:ext for sep=0}) and $\W_0$ be the circuit implied by~(\ref{lemma:minimal conditions:ext=0}).  Take $\Copies$ to a large enough value guaranteeing that $\W_0$ has an internal copy $\T$ in $\Copies\Unroll\Syn$ such that the minimum $x$ and $y$ for which $\T$ has vertices at rows $x$ and $\Copies-y$ are also large enough.  For every $i \in \Range{\Copies}$, write $A_i$ to denote the $i$-th copy of some $A \in \W_N$, and let $e = |\Ext(\W_N)|$.  Fix $k$ with $y \ll k \ll \Copies$, and consider any $0 \ll q \ll x$.  By following $k$ copies of $\W_N$ from $A_q$, we obtain a walk $\T_q$ of $\Copies\Unroll\Syn$ that goes through $A_{ie+q}$ for every $i \in \Range{k}$ (\Figure~\ref{fig:minimal conditions:ext=-1}(a)).  By Observation~\ref{obs:walk drawing} and Theorem~\ref{thm:plane drawing}, $\T_q$ and $\T$ share some vertex $T_q$ that belongs to the subpath of $\T_q$ that begins at $A_{ae+q}$ and ends at $A_{(a+1)e+q}$.  Note that, since $q \ll x$, then $T_{e+q} = T_q$, thus there exists a combination of $q$ and $a$ such that $T_{q+1}$ belongs to the subpath of $\T_{q+1}$ that begins at $A_{q+1}$ and ends at $A_{ae+q+1}$ (\Figure~\ref{fig:minimal conditions:ext=-1}(a)).  Then, the walk $\T'$ obtained by traversing $\T_q$ from $A_{ae+q}$ to $T_q$, then $\T$ from $T_q$ to $T_{q+1}$, and finally $\T_{q+1}$ from $T_{q+1}$ to $A_{ae+q+1}$ has $\Ext = -1$ (\Figure~\ref{fig:minimal conditions:ext=-1}(b)).  Furthermore, as $\W_0$ and $\W_N$ are circuits with maximum $\Sep_{\Desc}$ by Theorem~\ref{thm:separation constraints}, the circuit $\W_1$ of $\Syn$ that has $\T'$ as its copy has maximum $\Sep_{\Desc}$ as well, i.e., $\Sep_{\Desc}(\W_1) = 0$.
\end{proof}

If the circuit $\W_0$ of the previous lemma is internal, then $\Len$ is an integer combination of $\Dist$ and $\BegDist$ by~\eqref{eq:sep} and Corollary~\ref{cor:one more hollow}.  The following technical lemma is to deal with the otherwise case.  

\begin{lemma}\label{lem:rolled up circuit}
 For every UCA descriptor $\Desc$ and every circuit $\W$ of $\Syn$, there exists a circuit $\W'$ with $\Sep_\Desc(\W') \geq \Sep_\Desc(\W)$ and $\Ext(\W') = \Ext(\W)$ that has an internal $\T'$ copy in $\Copies\Unroll\Syn$ such that $\T'$ is either a path (if $\Ext(\W') \neq 0$) or a cycle (otherwise), for $\Copies = (3(|\Ext(\W)|+1)n)$.
\end{lemma}

\begin{proof}
 Let $e = \Ext(\W)$, $s = |e|/e$, $\Delta = |e|n$ and $\C_0, \ldots, \C_{k-1}$ be a partition of $\W$ into cycles.  Clearly, each $\C_i$ has an internal copy $\T_i$ in $(3n)\Unroll\Syn$, because $|\C_i| \leq n$ for $i\in \Range{k}$.  Let $\delta_i = s(y_i - x_i)$ where $x_i$ and $y_i$ are the lowest and highest rows reached by $\T_i$, respectively.  Since $\delta_i \in \Range{-n,n}$ and $\sum_{i\in\Range{k}} \delta_i \in \Range{-\Delta, \Delta}$, there exists a permutation $\pi$ of $\Range{k}$ such that $\sum_{i\in\Range{j}} \delta_{\pi(i)} \in \Range{-\Delta - n, \Delta + n}$ for every $j \in \Range{k}$.  This means that the walk $\T_\pi = \T_{\pi(1)}, \ldots, \T_{\pi(k)}$ of $\Copies\Unroll\Syn$ that begins in a vertex at row $r \in \Range{(e+1)n,2(e+1)n}$ is an internal copy of the circuit $\C_\pi = \C_{\pi(1)}, \ldots, \C_{\pi(k)}$ of $\Syn$.  By definition, $\Ext(\C_\pi) = 0$ and $\Sep_{\Desc}(\C_\pi) = \Sep_{\Desc}(\W)$.  Finally, observe that $\T_\pi$ can be partitioned into at most one path or cycle $\T'$ plus a family of cycles.  By construction, $\T'$ is internal in $\Copies\Unroll\Syn$ and it is the copy of some circuit $\W'$ with $\Ext(\W') = \Ext(\W)$.  Moreover, since every cycle of $\Copies\Unroll\Syn$ has $\Sep_\Desc \leq 0$, it follows that $\Sep(\W') \geq \Sep(\W)$ as desired.
\end{proof}

Now we are ready to state the main theorem of this section.

\begin{theorem}\label{thm:len and circ are integer}
  If\/ $\Unit$ is a $(\Dist, \BegDist)$-minimal $(\Circ, \Len, \Dist, \BegDist)$-CA model, then $\Len$ and $\Circ$ are integer combinations of $\Dist$ and $\BegDist$.
\end{theorem}

\begin{proof}
  By Lemma~\ref{lemma:minimal conditions}~(\ref{lemma:minimal conditions:ext=0}) and Lemma~\ref{lem:rolled up circuit}, $(3n)\Unroll\Syn$ contains an internal cycle $\T_0$ that is a copy of a circuit $\W_0$ of $\Syn$ with $\Ext(\W_0) = \Sep_\Desc(\W_0) = 0$.  Note that $\Sep_{\Desc}(\T_0) = \Sep_{\Desc}(\W_0) = 0$ because $\Ext(\W_0) = 0$.  Similarly, by Lemma~\ref{lemma:minimal conditions}~(\ref{lemma:minimal conditions:ext=-1}), $\Syn$ contains a circuit $\W_1$ with $\Ext(\W_1) = -1$ and $\Sep_\Desc(\W_1) = 0$.  Then, by~\eqref{eq:sep} and Corollary~\ref{cor:one more hollow},
  \begin{align}
    0 = \Sep_{\Desc}(\T_0) = \Len\Jump_0 + \Dist|\T_0|+\BegDist\sigma(\T_0) = -\Len + \Dist|\T_0|+\BegDist\sigma(\T_0) \mbox{, thus} \label{eq:len integrality} \\
    0 = \Sep_{\Desc}(\W_1) = -\Circ + d(|\T_0|\Jump_1 + |\W_1|) + \BegDist(\Jump_1\sigma(\T_0) + \sigma(\W_1)), \label{eq:circ integrality}
  \end{align}%
  where $\Jump_i = \Jump(\W_i)$ and $\sigma$ counts the number of steps in a walk.
\end{proof}

\subsection{Computing a minimal UCA model}
\label{sec:minimal algorithm}

Theorem~\ref{thm:len and circ are integer} yields an algorithm to compute a $(\Dist,\BegDist)$-minimal $\Desc$-CA model $\Unit$ equivalent to an input UCA model $\M$, when $\Dist \in \mathbb{Q}_{> 0}$, and $\BegDist \in \mathbb{Q}_{\geq 0}$ are also given as input.  The algorithm has three phases that compute $\Len$, $\Circ$, and $\Unit$, respectively.  For the first phase, recall that $\Syn$ has a circuit $\W_0$ with $\Sep_\Desc = \Ext = 0$ that has an internal copy $\T_0$ in $(3n)\Unroll\Syn$.  Moreover, $\T_0$ is a cycle and, by~\eqref{eq:len integrality}, $\Len = \Dist|\T_0| + \BegDist\sigma(\T_0)$.  Then, taking into account that every internal cycle of $(3n)\Unroll\Syn$ is a copy of a circuit of $\Syn$, we obtain that
\begin{equation}
  \Len = \max\{\Dist|\T| + \BegDist\sigma(\T) \mid \T \text{ is an internal cycle of }(3n)\Unroll\Syn\}\label{eq:minimal len}
\end{equation}
by Theorem~\ref{thm:separation constraints} and Corollary~\ref{cor:one more hollow}.  Graph $(3n)\Unroll\Syn$ has $O(n^2)$ vertices and edges.  So, as discussed in~\cite{SoulignacJGAA2017a}, the value of $\Len$ satisfying~\eqref{eq:minimal len} can be found in $O(n^3)$ time and $O(n^2)$ space.  Indeed, all we have to do is to compute $(3n)\Unroll\Syn$ to find the longest path in $(3n)\Unroll\Syn$ from the leftmost vertex at row $r$ to the rightmost vertex at row $r$, for each of the $O(n^2)$ rows $r$ of $(3n)\Unroll\Syn$.

\begin{lemma}\label{lem:len procedure}
  Let $\M$ be a UCA model equivalent to a $(\Dist, \BegDist)$-minimal $(\Circ, \Len, \Dist, \BegDist)$-CA model, for $\Dist\in \mathbb{Q}_{> 0}$ and $\BegDist \in \mathbb{Q}_{\geq 0}$.  There is an algorithm that computes $\Len$ in $O(n^3)$ time and $O(n^2)$ space when $\Unit$, $\Dist$ and $\BegDist$ are given as input. 
\end{lemma}

The second phase begins once the value of $\Len$ has been found.  Since $\Dist, \BegDist \in \mathbb{Q}$, we may write $\Dist = \frac{a_1}{b}$ and $\BegDist = \frac{a_2}{b}$ for $a_1, a_2, b \in \mathbb{N}$.  Let $\T_0$ and $\W_1$ be as in Theorem~\ref{thm:len and circ are integer} where, by Lemma~\ref{lem:rolled up circuit}, we may assume that $\W_1$ has an internal copy $\T_1$ in $(6n)\Unroll\Syn$.  Obviously, $|\T_0| \leq 3n^2$ and $|\T_1| \leq 6n^2$, thus $|\W_1|\leq 6n^2$.  Then, by~\eqref{eq:circ integrality}, taking into account that $\Jump(\W_1) \leq n$, it follows that $\Circ = \frac{a}{b}$ for some $a \leq k = 10(a_1+a_2)n^3$.  The idea, then, is to search for $\Circ \in [0, kn^3)$ with a bisection algorithm.  At each step, we test whether $\M$ is equivalent to a $\Desc=(\Circ', \Len, \Dist, \BegDist)$ by invoking the algorithm in Theorem~\ref{thm:separation constraints}.  This algorithm returns either a $\Desc$-CA model equivalent to $\M$ or a cycle $\W$ with $\Sep_{\Desc} > 0$.  In the former case, $\Circ \leq \Circ'$ by definition.  In the latter case, by~\eqref{eq:sep}, either $\Circ \leq \Circ'$ (if $\Ext(\W) \geq 0$) or $\Circ \geq \Circ'$ (if $\Ext(\W) \leq 0$).  This implies that only $O(\log n)$ invocations to the algorithm in Theorem~\ref{thm:separation constraints} are required in the bisection algorithm to locate $\Circ$.

\begin{lemma}\label{lem:circ procedure}
  Let $\M$ be a UCA model equivalent to a $(\Dist, \BegDist)$-minimal $(\Circ, \Len, \Dist, \BegDist)$-CA model, for $\Dist\in \mathbb{Q}_{> 0}$ and $\BegDist \in \mathbb{Q}_{\geq 0}$.  There is an algorithm that computes $\Circ$ in $O(n^2\log(n))$ time and $O(n)$ space when $\M$, $\Dist$, $\BegDist$, and $\Len$ are given as input. 
\end{lemma}

Finally, for the last phase, we simply invoke the algorithm in Theorem~\ref{thm:separation constraints} with $\Len$ and $\Circ$ as input.  Since this last steps costs $O(n^2)$ time and $O(n)$ space, we obtain the main theorem of this section. 

\begin{theorem}\label{thm:algorithm}
 Given a UCA model $\M$ and two values $\Dist\in \mathbb{Q}_{> 0}$ and $\BegDist \in \mathbb{Q}_{\geq 0}$, a $(\Dist, \BegDist)$-minimal $\Desc$-CA model can be computed in $O(n^3)$ time and $O(n^2)$ space.
\end{theorem}

\section{The minimum representation problem}
\label{sec:minimum representation}

In this section we prove that the minimum representation problem is well defined and \NP-complete.  In order to do so, we first review the structure of PCA models.  Our review is just a translation of some results by Huang~\cite{HuangJCTSB1995} to the framework of synthetic graphs; these results also appear in~\cite{KoeblerKuhnertVerbitskyDAM2017,SoulignacA2015}.

\begin{figure}[t!]
  \centering
  \begin{tabular}{c@{\hspace{1cm}}c@{\hspace{1cm}}c}
    \includegraphics{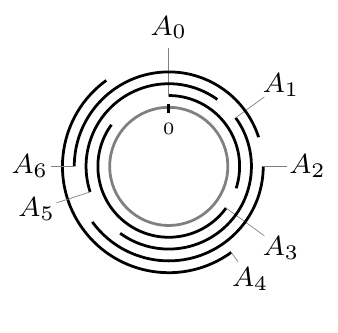} & \includegraphics{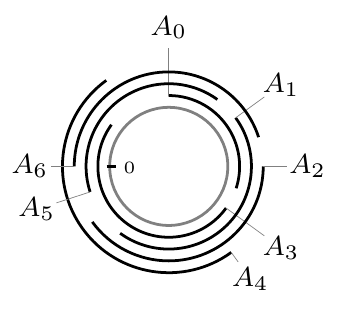} & \includegraphics{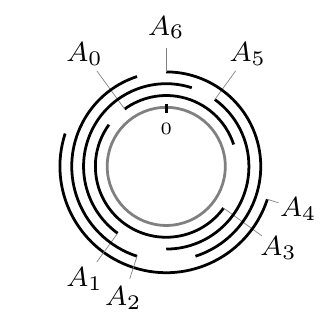} \\
    $\M$ & $\M'$ & $\M^{-1}$
  \end{tabular}
  \caption{An aligned PCA model $\M$ shift equivalent to $\M'$ whose reverse is $\M^{-1}$.}\label{fig:shift and reverse}
\end{figure}

Let $\M = (C, \A)$ be a PCA model and say that $A \in \A$ is \emph{universal} when $A$ intersects every arc in $\A$.  If no arc of $\M$ is universal, then $\M$ is \emph{universal-free}, while if all the arcs of $\M$ are universal, then $\M$ is \emph{complete}.  Any model $\M'$ obtained after moving the point $0$ of $C$ as being \emph{shift equivalent} to $\M$ (\Figure~\ref{fig:shift and reverse}(b)).  Let $\A^{-1} = \{A^{-1} \mid A \in \A\}$, where $A^{-1} = (|C|-t(A), |C|-s(A))$ for $A \in \A$.  (Sometimes we write $\M^1$ and $A^1$ to refer to $\M$ and $A$, respectively.)  The \emph{reverse} $\M^{-1}$ of $\M$ is PCA model obtained from $(C, \A^{-1})$ after moving $0$ to $t(A)$, where $A$ is the last arc of $\A$ w.r.t $<$ (\Figure~\ref{fig:shift and reverse}(c)). Clearly, $\M$ is a $(\Circ,\Len)$-CA model if and only if $\M^{-1}$ is a $(\Circ,\Len)$-CA model, while $A_0 < \ldots < A_{n-1}$ are the arcs of $\M$ if and only if $A_{n-1}^{-1} < \ldots < A_{0}^{-1}$ are the arcs of $\M^{-1}$.   If $\M^{-1}$ is the unique PCA model isomorphic to $\M$, up to shift equivalence, then $\M$ is called \emph{singular}. 

\begin{figure}[t!]
  \centering
  \begin{tabular}{c@{\hspace{5mm}}c@{\hspace{5mm}}c}
    \includegraphics{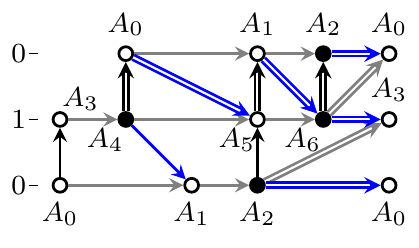} & \includegraphics{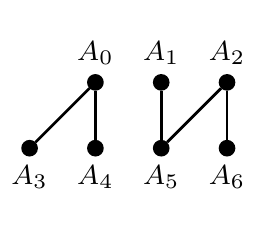} & \includegraphics{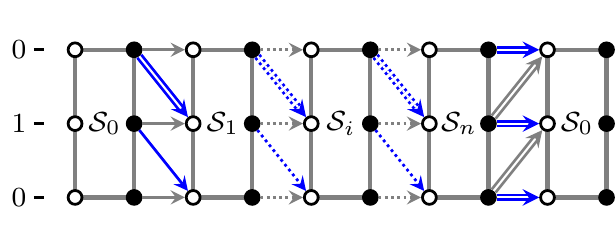} \\
    (a) & (b) & (c)
  \end{tabular}
  \caption{(a) Mitas' drawing $\S(\M)$ of $\M$ (\Figure~\ref{fig:shift and reverse}) with backward and external edges; co-starts are filled with white.  Note that $A_0$ and all the black nodes are co-ends.  (b) The complement of the intersection graph of $\M$. (c) The general case when $\M$ is universal-free.}\label{fig:boundary}
\end{figure}

A greedy nose cycle $\W$ of $\Syn$ is a \emph{boundary} when it contains exactly two noses.  If $B \to A$ is a nose of $\W$, then $A$ is a \emph{co-start}, while the arc immediately preceding $A$ in the circular ordering implied by $<$ is a \emph{co-end}.  Clearly, every boundary has two co-starts, thus $\M$ has an even number of co-start and co-end arcs.  We say that $\M$ is \emph{aligned} when its first arc (w.r.t $<$) is a co-start.  It is easy to see that $\Syn$ has exactly two rows when $\M$ is aligned (\Figure~\ref{fig:boundary}(a)).  

Let $G$ be the complement of the intersection graph of $\M$, i.e., $G$ has a vertex $v_A$ for each $A \in \A$, while $v_A$ and $v_B$ are adjacent if and only if $A \cap B = \emptyset$ (\Figure~\ref{fig:boundary}(b)).  We say that a subset $\A'$ of $\A$ \emph{induces a co-component} of $\M$ when their corresponding vertices induce a component in $G$.  Sometimes we also refer to the PCA model $(C, \A')$ as being a co-component of $\M$.  The next results describe the structure of a PCA model (\Figure~\ref{fig:boundary}).

\begin{lemma}[\cite{HuangJCTSB1995,KoeblerKuhnertVerbitskyDAM2017,SoulignacA2015}]\label{lem:singular}
 If $\M$ is nonsingular and universal-free, then $\M$ has $k>1$ co-components and $\Syn$ has $k$ boundary cycles.
\end{lemma}

\begin{theorem}[\cite{HuangJCTSB1995,KoeblerKuhnertVerbitskyDAM2017,SoulignacA2015}]\label{thm:co-components}
 Suppose $\M$ is an aligned and universal-free PCA model with co-starts $A_0 < \ldots < A_{2k-1}$ ($k>0$), and let $B_0 < \ldots < B_{2k-1}$ be its co-ends. Define $\A_i^j$ as the family of arcs between $A_{kj+i}$ and $B_{kj+i}$ for $i \in \Range{k}$, $j \in \{0,1\}$, and $\A_i = \A_i^{0} \cup \A_i^{1}$.  If we write $A_{(2+j)k} = A_{jk}$, then:
 \begin{enumerate}[(a)]
  \item The steps $B_{jk+i} \to A_{jk+i+1}$ and hollows $B_{jk+i} \to A_{(1-j)k+i+1}$ are the only edges of $\Syn$ that can go from a vertex in $\A_i$ to a vertex outside $\A_i$.\label{thm:co-components:bridges}
  \item The submodel $\M_i = (C,\A_i)$ is a singular and aligned co-component of $\M$.\label{thm:co-components:submodel}
  \item $\Syn(\M_i)$ is the $3$-digraph that is obtained after inserting the steps $B_{jk+i} \to A_{(1-j)k+i}$ and the hollows $B_{jk+i} \to A_{jk+i}$ to the sub-$3$-digraph of $\Syn$ induced by $\A_i$.\label{thm:co-components:subsyn}
 \end{enumerate}
\end{theorem}

By Theorem~\ref{thm:co-components}, we write that $\M_0 < \ldots < \M_{k-1}$ are the co-components of an aligned and universal-free model $\M$ to indicate that $\M_i$ has $A_i$ as its first arc, where $A_0 < \ldots < A_{2k-1}$ are the co-starts of $\M$.  

For a given model $\M$, we can compute the minimal value of $\Circ$ from the minimal value of $\Len$ by applying the increasing function defined by~\eqref{eq:circ integrality} to a circuit of $\Syn$ with $\Ext=-1$ and $\Sep=0$.  To prove the existence of minimum models, we show that some of these circuits belong to all the PCA models isomorphic to $\M$, because they are ``trapped'' inside singular submodels of $\M$.  The next lemma describes such circuits, while the following Theorem completes the proof.

\begin{lemma}\label{lem:paths in co-components}
  Let\/ $\Unit$ be a $(\Dist,\BegDist)$-minimal $\Desc$-CA model.  If\/ $\Unit$ is a\-ligned, then\/ $\Syn$ contains a circuit with\/ $\Sep_\Desc = 0$ and $\Ext = -1$ that has a forward copy in $(6n) \Unroll\Syn$.
\end{lemma}

\begin{proof}
  For $i \in \Range{12n}$, let $A_i$ and $B_i$ be the leftmost and rightmost arcs at row $i$ of $(6n)\Unroll\Syn$, respectively.  We remark that $A_{2r+i}$ is a copy of $A_i$ for every $r \in \Range{6n}$ and $i \in \{0,1\}$ because $\Unit$ is aligned.  Thus, every even (resp.\ odd) row of $(6n)\Unroll\Syn$ is a copy of row $0$ (resp.\ $1$).  By Lemmas \ref{lemma:minimal conditions} (\ref{lemma:minimal conditions:ext=-1})~and~\ref{lem:rolled up circuit}, $\Syn$ has a circuit $\W$ with $\Sep_\Desc = 0$, $\Ext = -1$, and an internal copy $\T$ in $(6n)\Unroll\Syn$ that is a path.  Suppose $\T$ is not forward and note that, by taking an appropriate starting vertex, we may assume that $\T$ goes from $A_i$ to $A_{i+2}$ for some $i \in \Range{12n-2}$ (\Figure~\ref{fig:paths in co-components}(a)).  Let $j$ be the lowest row reached by $\T$ and, among all the possibilities for $\W$, take one such that $j$ is as large as possible and $\T$ has the fewest number of arcs at row $j$.  
  
  Let $X_j$ be the maximum arc (w.r.t. $<$) at row $j$ that belongs to $\T$.  Since $\T$ is a path, Theorem~\ref{thm:plane drawing} and Observation~\ref{obs:walk drawing} imply that the edge $X_j \to Y_{j+1}$ of $\T$ is not a hollow.  The edge $X_{j} \to Y_{j+1}$ is neither a step; otherwise $X_j = B_i$, $Y_{j+1} = A_{i+1}$, and the circuit obtained by replacing $B_i \to A_{i+1}$ with the hollow $B_i \to A_i$ plus the nose path from $A_i$ to $A_{i+1}$ would have $\Sep_\Desc > 0$, contradicting Theorem~\ref{thm:separation constraints}.  Thus, $X_j \to Y_{j+1}$ is a $1$-nose (\Figure~\ref{fig:paths in co-components}(a)).
  
  We claim that $X_j$ is not universal; suppose otherwise to obtain a contradiction.  Since $X_j \to Y_{j+1}$ is a nose, the copy $X_{j+2}$ of $X_j$ at row $j+2$ has a hollow to $Y_{j+1}$.  Moreover, if $Z_{j+1} \to Y_{j+1}$ is a step, then $Z_{j+1}$ has a nose to an arc at row $j+2$ and a hollow to its copy at row $j$.  Thus, the unique nose path $\T_N$ from $Z_{j+1}$ to $X_{j+2}$ has the same length as the unique hollow path $\T_H$ from $Z_{j+1}$ to $X_j$.  Note that $\T$ must contain $\T_H$ as a subwalk because, by the minimality of $j$, it cannot contain any nose ending at row $j$.  But then, we can replace $\T_H, Y_{j+1}$ with $\T_N, Y_{j+1}$ in $\T$ to obtain a walk whose lowest row is as at least $j$ and that has fewer arcs than $\T$ at row $j$.  Therefore, $X_j$ is not universal.  
  
  Since $X_j$ is not universal in $\Unit$, there exists a nose path from $Y_{j+1}$ to $X_{j+2}$.  Hence, there exists a nose $L_{j+1} \to L_{j+2}$ with $Y_{j+1} \leq L_{j+1}$ and $L_{j+2} \leq X_{j+2}$. Among all the possibilities, take $L_{j+2}$ to be maximum and let $R_{j+2} \to R_{j+1}$ be the hollow with minimum $R_{j+2} \geq L_{j+2}$.  Note that either $R_{j+2}$ is the rightmost vertex at row $j+2$ or some nose ends at the arc that immediately follows $R_{j+2}$. Whichever the case, $L_{j+2} \leq X_{j+2} \leq R_{j+2}$ by the maximality of $L_{j+2}$ (\Figure~\ref{fig:paths in co-components}(a)).  Note also that $L_{j+1} \to R_{j+1}$ is a step because of the minimality of $R_{j+2}$ (\Figure~\ref{fig:paths in co-components}(a)).  By Theorem~\ref{thm:separation constraints}, $\T$ cannot contain the step $L_{j+1} \to R_{j+1}$ because we could replace it with the path from $L_{j+1}$ to $R_{j+1}$ that contains the steps between $L_{j+2}$ and $R_{j+2}$, obtaining a circuit with $\Sep_\Desc > 0$ (\Figure~\ref{fig:paths in co-components}(b)).
   
  Let $\T'$ be the copy of $\T$ (possibly in $(6n+2)\Unroll\Syn$) that beings at $A_{i+2}$ (\Figure~\ref{fig:paths in co-components}(c)).  By Theorem~\ref{thm:plane drawing} and Observation~\ref{obs:walk drawing}, it can be observed that $\T$ and $\T'$ share some vertex $T$ in the subpath of $\T'$ that goes from $A_{i+2}$ to $X_{j+2}$ (Fig~\ref{fig:paths in co-components}(c)).  Therefore, the subpath from the copy of $T$ at row $\Row(T) - 2$ to $T$ is forward in $(6n)\Unroll\Syn$, and has $\Ext = -1$ and $\Sep_\Desc = 0$, as desired.
\end{proof}

\begin{figure}[t!]
  \centering
  \begin{tabular}{c@{\hspace{5mm}}c@{\hspace{5mm}}c}
    \includegraphics{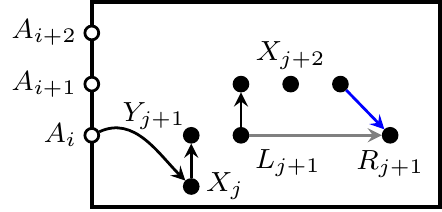} & \includegraphics{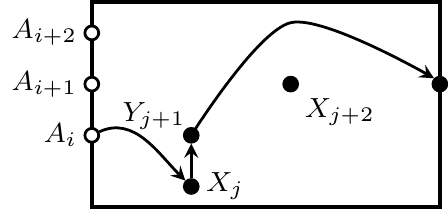} & \includegraphics{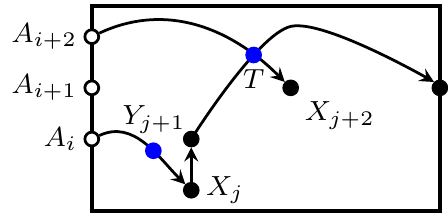} \\
    (a) & (b) & (c)
  \end{tabular}
  \caption{Proof of Lemma~\ref{lem:paths in co-components}.}\label{fig:paths in co-components}
\end{figure}

\begin{theorem}
  Every $(\Circ,\Len,\Dist,\BegDist)$-CA model $\M$ is equivalent to a $(\Dist,\BegDist)$-minimum UCA model.
\end{theorem}

\begin{proof}
  If $\M$ is complete and $\Delta = \Dist + \BegDist$, then the $(2(n-1)\Delta+2\Dist,(n-1)\Delta+\Dist,\Dist,\BegDist)$-CA model that has one beginning point at $i\Delta$ for every $i \in \Range{n}$ is a $(\Dist,\BegDist)$-minimum model isomorphic to $\M$.  If $\M$ is singular, then, as neither moving the position of $0$ nor reversing the arcs of $\M$ affects the length of the circle or the arcs, it follows that any $(\Dist,\BegDist)$-minimal model equivalent to $\M$ is also $(\Dist,\BegDist)$-minimum~\cite{SoulignacJGAA2017}.  For the remaining of the proof, suppose $\M$ is neither singular nor complete, and let $\Desc = (\Circ,\Len,\Dist,\BegDist)$.  

  Let $\Len^* \leq \Len$ be the minimum such that the UCA model $\M$ is isomorphic to a $(\Dist,\BegDist)$-minimal $v$-CA model $\M^*$, for $v=(\Circ^*, \Len^*,\Dist,\BegDist)$.  We may suppose that $\M$ and $\M^*$ are aligned because, as previously stated, the position of $0$ is irrelevant under isomorphism.  Let $\W$ be the circuit of $\Syn(\M^*)$ implied by Lemma~\ref{lem:paths in co-components}.  By Theorem~\ref{thm:co-components} (\ref{thm:co-components:bridges})--(\ref{thm:co-components:subsyn}), $\W$ is a circuit of $\Syn(\M_i)$, where $\M_i$ is a singular model that induces a co-component of $\M^*$.  By Theorem~\ref{thm:co-components} (\ref{thm:co-components:bridges})--(\ref{thm:co-components:subsyn}), taking into account that $\M_i$ is singular, it follows that $\M_i$ is a submodel of either $\M$ or $\M^{-1}$, while $\W$ is a circuit of either $\Syn(\M)$ or $\Syn(\M^{-1})$.  Clearly, $\Jump(\W) = 2$, hence, by \eqref{eq:sep} and Theorem~\ref{thm:separation constraints}, it follows that
  \begin{align*}
   0 \geq \Sep_\Desc(\W) - \Sep_v(\W) = \Circ^*-\Circ + 2(\Len-\Len^*),
  \end{align*}
  thus $\Circ^* \leq \Circ$ as well.  Therefore, $\M^*$ is minimum.
\end{proof}

\subsection{Computing a minimum UCA model}

A PCA model $\M$ can be isomorphic to an exponential number of PCA models that are pairwise not shift equivalent.  These models arise from joining the co-components of $\M$ via three well defined operations.  Again, these operations are described in~\cite{HuangJCTSB1995,KoeblerKuhnertVerbitskyDAM2017,SoulignacA2015}; here we just translate them to the language of synthetic graphs.  For the sake of exposition, we describe the effects of these operations using both PCA models and synthetic graphs.

Let $\M$ (resp.\ $\M'$) be an aligned and universal-free PCA model whose synthetic graph $\Syn$ (resp.\ $\Syn'$) has $A_r$ and $B_r$ (resp.\ $A_r'$ and $B_r'$) as its leftmost and rightmost arcs in row $r$, for $r\in\{0,1\}$.  The \emph{$i$-alignment} of $\M$ is the model $\M|i$ obtained by placing $0$ at the position of $s(A_i)$.  Obviously, $\Syn|i = \Syn(\M|i)$ is obtained from $\Syn$ by exchanging rows $0$ and $i$.  The \emph{join} $\Syn + \Syn'$ is the synthetic graph aligned at $A_0$ that is obtained from $\Syn \cup \Syn'$ by replacing: the steps $B_i \to A_{1-i}$ and $B_i' \to A'_{1-i}$ with the steps $B_i \to A_{1-i}'$ and $B_i' \to A_{1-i}$; and the hollows $B_i \to A_i$ and $B_i' \to A_i'$ with the hollows $B_i \to A_i'$ and $B_i' \to A_i$.  (The removed edges exist by~Theorem~\ref{thm:co-components} (\ref{thm:co-components:bridges}).)  The \emph{join} $\M + \M'$ is the unique PCA model whose synthetic graph is $\Syn+\Syn'$.  Informally, we can build $\M + \M'$ directly from $\M$ and $\M'$ by first cutting the circles of $\M$ and $\M'$ immediately before $t(A_i)$ and $t(A_i')$, respectively, and then gluing the four obtained lines into a circle where $s(A_0) < s(A_0') < s(A_1) < s(A_1')$ (\Figure~\ref{fig:join}).

\begin{figure}
 \begin{tabular}{c@{\hspace{0mm}}c@{\hspace{0mm}}c@{\hspace{5mm}}c}
  \includegraphics{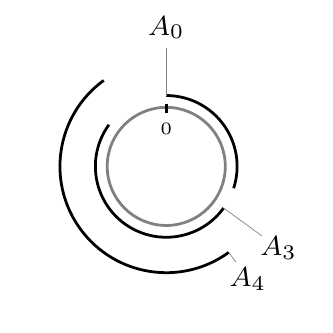} & \includegraphics{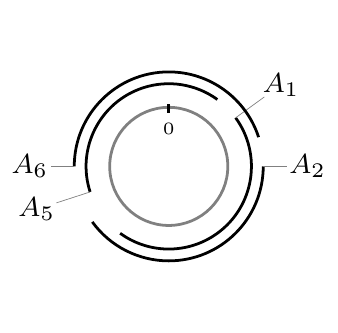} & \includegraphics{fig-model-order} & \includegraphics{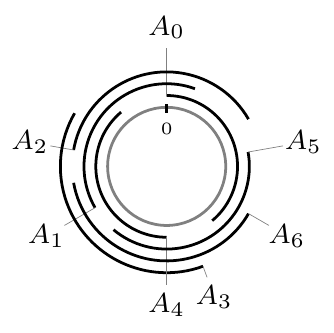} \\
  \includegraphics{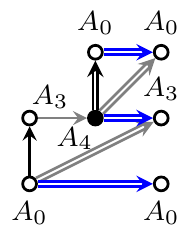} & \includegraphics{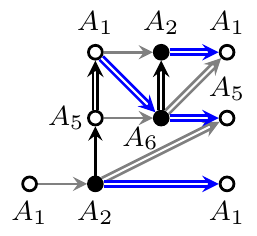} & \includegraphics{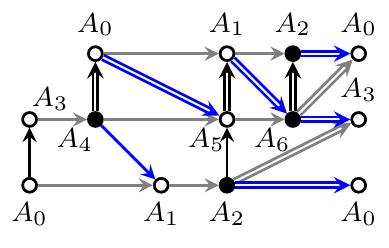} & \includegraphics{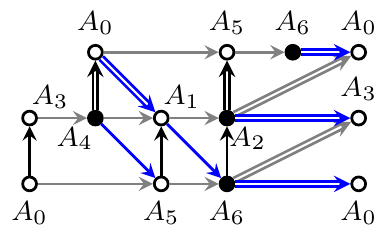}
 \end{tabular}
 \caption{Top row from left to right: $\M_1$, $\M_2$, $\M_1 + \M_2$, $\M_1 + \M_2|1$.  Bottom row: the corresponding synthetic graphs.}\label{fig:join}
\end{figure}

Let $\M_0 < \ldots < \M_{k-1}$ be the co-components of an aligned and universal-free PCA model $\M$.  Obviously, $\M_i = \M_i|0$ for every $i \in \Range{k}$, while, by Theorem~\ref{thm:co-components}, $\M = \M_0|0 + \ldots + \M_{k-1}|0$.  If we exchange the order of the co-components in the summation, or if we replace $\M_i$ with $\M_i^{-1}$ or $\M_i|0$ with $\M_i|1$, then we can obtain a PCA model that is not shift equivalent to $\M$ (\Figure~\ref{fig:join}).  In fact, as it was observed by Huang~\cite{HuangJCTSB1995}, all the PCA models isomorphic to $\M$, up to shift equivalence, can be obtained in this way.  

To state Huang's result in formal terms, consider the permutation $\pi$ of $\Range{k}$ and the functions $\varphi\colon \Range{k} \to \{-1,1\}$ and $\psi\colon \Range{k} \to \{0,1\}$. Define:
\begin{displaymath}
 \pi(\psi(\varphi(\M))) = \left.\M_{\pi(0)}^{\varphi(0)}\right\vert\psi(0) + \ldots + \left.\M_{\pi(k-1)}^{\varphi(k-1)}\right\vert\psi(k-1).
\end{displaymath}
As stated before, by Theorem~\ref{thm:co-components}, $\M = \pi(\psi(\varphi(\M)))$ for the \emph{identity} mappings $\varphi = 1$, $\psi = 0$, and $\pi(i) = i$.  Here, $\pi$ is used to permute the co-components of the summation, $\varphi$ selects between a co-component and its reverse, and $\psi$ defines the alignment of the co-component.  Obviously, we can omit some of these functions from the notation if the corresponding identity is applied, e.g., $\pi(\M) = \pi(0(1(\M)))$.  The reader can check that the order between the operations is unimportant (assuming that $\psi$ and $\varphi$ are modified according to $\pi$; see~\cite{HuangJCTSB1995,KoeblerKuhnertVerbitskyDAM2017,SoulignacA2015}).

\begin{theorem}[\cite{HuangJCTSB1995,KoeblerKuhnertVerbitskyDAM2017,SoulignacA2015}]\label{thm:all models}
 Two aligned and universal-free PCA models $\M$ and $\M'$ with $k$ co-components are isomorphic if and only if $\M'$ is equivalent to $\pi(\psi(\varphi(\M)))$ for some permutation $\pi$ of $\Range{k}$, and $\varphi\colon \Range{k} \to \{-1,1\}$ and $\psi\colon \Range{k} \to \{0,1\}$
\end{theorem}

Before dealing with the minimization problem, we prove that an auxiliary weighing problem is strongly \NP-complete.  In a \emph{weighing problem}, the goal is to weigh the edges of a digraph $G$, obeying certain rules, to minimize the maximum among the weights of the cycles in $G$.  In our problem, $G$ has a vertex $v_q^r$ for every $q\in \Range{2k}$ and $r \in \Range{4k}$, one \emph{diagonal} edge $v_q^r \to v_{q+1}^{r-1}$, one \emph{vertical} edge $v^{r-1}_q \to v^r_q$, and one \emph{horizontal} edge $v^r_{2k-1} \to v^r_0$; of course, the edges are present when $r > 1$ and $q+1 < 2k$ (\Figure~\ref{fig:grid-graph}).  Following \Figure~\ref{fig:grid-graph}, we say $G$ has $4k$ \emph{rows} and $2k$ \emph{columns}, whereas $v_q^r$ is at \emph{row} $r$ and \emph{column} $q$.  

\begin{figure}[t!]
 \centering
 \includegraphics{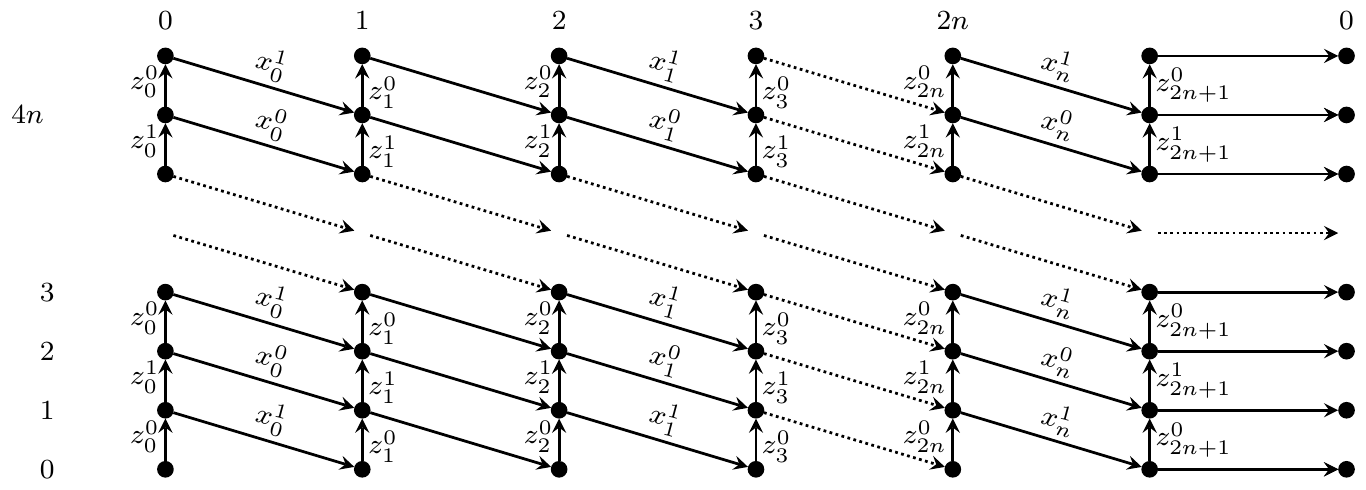}
 \caption{Graph of our auxiliary weighing problem, where $n=k-1$, $z_{2q}^p = y_q^p + x_q^p + 1$, and $z_{2q+1}^p = y_q^{1-p}+x_q^p+1$ for $q \in \Range{k}$ and $p \in \{0,1\}$.}\label{fig:grid-graph}
\end{figure}

Each possible weighing $\Wg_X$ of $G$, in turn, is defined by a sequence of tuples
\begin{displaymath}
 X = (x_0^0, x_0^1, y_0^0, y_0^1), \ldots, (x_{k-1}^0, x_{k-1}^1, y_{k-1}^0, y_{k-1}^1)
\end{displaymath}
with $x_q^p, y_q^p \in \mathbb{N}_{>0}$ for $q\in\Range{k}, p\in\{0,1\}$.  As depicted in Figure~\ref{fig:grid-graph}, if $p = r \bmod 2$, then
\begin{itemize}
 \item $\Wg_X(v^r_{2q} \to v^{r-1}_{2q+1}) = x_q^p$,
 \item $\Wg_X(v^r_{2q+1} \to v^{r-1}_{2q+2}) = \Wg_X(v^r_{2k-1} \to v^r_{0}) = 1$ for $q\in\Range{k-1}$,
 \item $\Wg_X(v^r_{2q} \to v^{r+1}_{2q}) = y_q^p + x_q^p + 1$, and
 \item $\Wg_X(v^r_{2q+1} \to v^{r+1}_{2q+1}) = y_q^{1-p} + x_q^p + 1$.
\end{itemize}

We consider three operations on $X$ that are defined by $\chi,\gamma\colon \Range{k} \to \{0,1\}$ and a permutation $\pi$ of $\Range{k}$, as shown in \Figure~\ref{fig:swaps}. Formally, $\chi(X)$ (resp.\ $\gamma(X)$) is obtained from $X$ by swapping $x_i^0$ and $x_i^1$ (resp.\ $y_i^0$ and $y_i^1$) when $\chi(i) = 1$ (resp.\ $\gamma(i) = 1$), while $\pi(X)$ is the sequence obtained by placing the $i$-th tuple of $X$ at position $\pi(i)$. 

With all these ingredients, we can now formulate our auxiliary weighing problem.  For the sake of notation, we write $\Wg_X(G)$ to denote the weight of the maximum cycle of $G$ when $\Wg_X$ is applied.
\begin{description}
 \item [Minimum cycle weighing by columns of a pseudo-grid digraph (MCW)]
 \item [Instance:] $X = \{(x_q^0, x_q^1, y_q^0, y_q^1) \mid q\in\Range{k} \text{ and } x_q^0,x_q^1,y_q^0,y_q^1 \in \mathbb{N}_{>0}\}$ and $\Len\in\mathbb{N}$
 \item [Question:] Is $\Wg_Y(G) \leq \Len$ for some $Y = \pi(\chi(\gamma(X)))$?
\end{description}

\begin{figure}
  \mbox{}\hfill \includegraphics{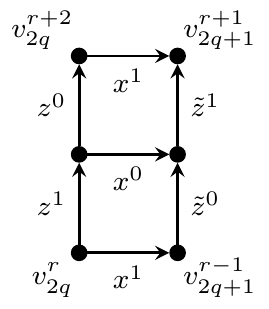} \hfill \includegraphics{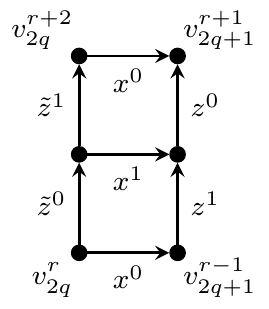} \hfill \includegraphics{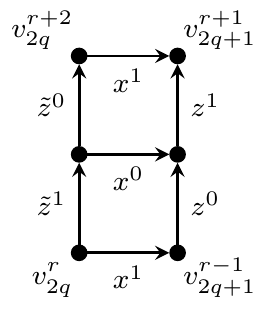} \hfill \includegraphics{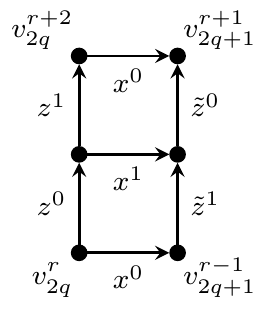} \hfill\mbox{}
  \caption{Section of $G$ showing $\Wg_Y$ for $Y = \pi(\chi(\gamma(X)))$.  Here $(x^0, x^1, y^0, y^1)$ is $i$-th tuple of $X$ such that $\pi(i) = q$, while $z^p = x^p + y^p + 1$ and $\tilde{z}^p = x^p + y^{1-p} + 1$ for $p \in \{0,1\}$.  From left to right, each figure depicts the case: $\chi(i) = \gamma(i) = 0$; $\chi(i) = 1$ and $\gamma(i) = 0$; $\chi(i) = 0$ and $\gamma(i) = 1$; and $\chi(i) = \gamma(i) = 1$.}\label{fig:swaps}
\end{figure}

\begin{theorem}\label{thm:MCW is NPC}
 MCW is strongly \NP-complete.
\end{theorem}

\begin{proof}
 As discussed in Section~\ref{sec:minimal algorithm}, every cycle of $G$ has an horizontal edge $v_{2k-1}^r \to v_0^r$, thus we can compute $\Wg_Y(G)$ in polynomial time by looking at the weights of the paths that go from $v_0^r$ to $v_{2k-1}^r$, for every $r \in \Range{4k}$.  Therefore, as any sequence $Y = \pi(\chi(\gamma(X)))$ with $\Wg_Y(G) \leq \Len$ serves as a certificate authenticating that $(X, \Len)$ is a yes instance of MCW, it follows that the minimum cycle weighing problem belongs to \NP.  To prove its hardness, we show a polynomial-time reduction from the $3$-partition problem that is known to be strong \NP-complete~\cite{GareyJohnson1979}.  
\begin{description}
 \item [$3$-partition]
 \item [Instance:] a set $S = \{s_0, \ldots, s_{3n-1}\}$ with $nT = \sum S$ (and $s \in\Range{1,T}$ for $s \in S$).
 \item [Question:] can $S$ be partitioned into sets $S_0, \ldots, S_{n-1}$ such that $|S_i| = 3$ and $\sum S_i = T$ for every $i \in \Range{n}$?
\end{description}

Consider an instance $S = \{s_0,\ldots, s_{3n-1}\}$ of the $3$-partition problem with $nT = \sum S$.  For $i \in \Range{n}$, let $l_i = 2(n^2+i)$, $h_i = l_i+2$, $y_\infty = h_n^2$, and:
\begin{itemize}
 \item $X_{3i+j} = (T, s_{3i+j}, 1, 1)$ for every $j \in \Range{3}$,
 \item $L_i = (l_iT, 1, y_\infty-l_iT -1, y_\infty-2)$, $H_i = (1, h_iT, y_\infty-h_iT-1, y_\infty-2)$, and
 \item $X = L_0, X_0, X_1, X_2, H_0, L_1, X_3, X_4, X_5, H_1, \ldots, L_{n-1}, X_{3n-3}, X_{3n-2}, X_{3n-1}, H_{n-1}$.
\end{itemize}
The graph $G$ weighted with $\Wg_X$ is shown in \Figure~\ref{fig:grid graph proof}.  Clearly, $X$ can be computed in polynomial time because the values in $X$ have a polynomial size with respect to those in $S$.  Thus, it suffices to prove that $S$ is a yes instance of the $3$-partition problem if and only if $(X, \Len)$ is a yes instance of MCW, for $\Len = (10n-1)y_\infty + \sum_{j\in\Range{n}} h_jT + n(T+6)$.  

\begin{figure}
 \centering \includegraphics{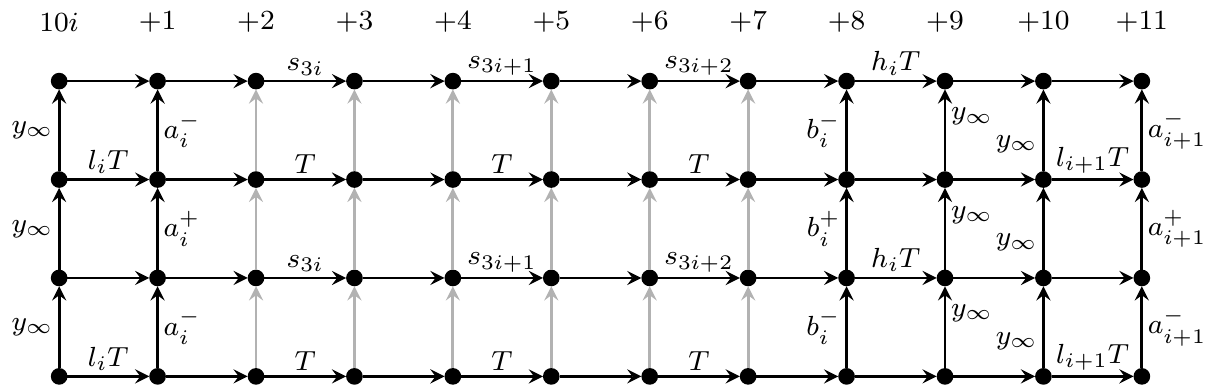}
 \caption{A section of the graph $G$ weighed with $\Wg_X$.  Here $a_i^+ = y_\infty + l_iT - 1$, $a_i^- = y_\infty -l_iT + 1$, $b_i^+ = a_i^+ + 2T$, $b_i^- = a_i^- - 2T$, and $\Wg_X(e) \leq T+2$ for every gray edge.}\label{fig:grid graph proof}
\end{figure}

Suppose first that $S_0, \ldots, S_{n-1}$ is a partition of $S$ with $|S_i| = 3$ and $\sum S_i = T$ where, w.l.o.g., $S_i = \{s_{3i}, s_{3i+1}, s_{3i+2}\}$ for every $i \in \Range{n}$.  To show that $\Wg_X(G) \leq \Len$, let $C$ be a cycle of $G$ with maximum weight that has a maximum number of vertical edges with $\Wg_X = y_\infty$.  We claim that all the vertical edges of $C$ have $\Wg_X = y_\infty$.  Indeed, observe that $C$ has $10n-1$ vertical edges, $10n-1$ diagonal edges with $\Wg_X \leq T$, and one horizontal edge (with $\Wg_X = 1$).  Therefore, $C$ cannot have vertical edges with weight at most $T+2$, as if it has $m > 0$ of such edges, then
\begin{displaymath}
 \Wg_X(C) < 10(n-m-1)(y_\infty+h_n) + m(T+2) + 10nT < (10n-1)y_\infty < \Wg_X(C'),
\end{displaymath}
where $C'$ is any cycle whose vertical edges all have $\Wg_X = y_\infty$.  Similarly, if $v_q^r \to v_q^{r+1}$ is the first (resp.\ last) vertical edge with $\Wg_X = y_\infty \pm (l_i+1)$ (resp.\ $\Wg_X = y_\infty \pm (h_i+1)$) for $i\in\Range{n}$, then we can replace the subpath $v_q^r, v_q^{r+1}, v_{q+1}^r$ (resp.\ $v_{q-1}^{r+1}, v_q^r, v_q^{r+1}$) of $C$ with the path $v_q^r, v_{q+1}^{r-1}, v_{q+1}^{r}$ (resp.\ $v_{q-1}^{r+1}, v_{q-1}^{r+2}, v_{q+1}^{r+1}$), obtaining a new cycle $C'$ with $\Wg_X(C') = \Wg_X(C)$ that has one more vertical edge with $\Wg_X = y_\infty$ (\Figure~\ref{fig:grid graph proof}).  

Since all the vertical edges of $C$ have $\Wg_X = y_\infty$, it follows that $C$ has a subpath $P_i$ from $v_{10i}^r$ to a vertex at column $10i+9$ that contains only diagonal edges, for every $i \in \Range{n}$ and some $r \in \Range{20n}$ (\Figure~\ref{fig:grid graph proof}). By construction, $\Wg_X(P_i) = l_iT + 3T + 5 = h_iT + T + 5$ when $r$ is even, while $\Wg_X(P_i) = h_iT + \sum S_i + 5 = h_iT + T + 5$ when $r$ is odd.  Therefore, taking into account those edges that belong to no $P_i$, we obtain
\begin{displaymath}
  \Wg_X(C) = (10n-1)y_\infty+\sum_{i \in \Range{n}} h_iT + nT + 6n = \Len.
\end{displaymath}
Summing up, $(X, \Len)$ is a yes instance of the MCW.

For the converse, suppose $(X, \Len)$ is a yes instance of MCW, i.e.,  there exist $\chi,\gamma \in \Range{5n} \to \{0,1\}$ and a permutation $\pi$ of $\Range{5n}$ such that $\Wg_Y(G) \leq \Len$, for $Y = \pi(\chi(\gamma(X)))$.  In the following, for any walk $W$ of $G$, we write $W^+$ to denote the walk of $G$ whose $i$-th vertex is $v^{r+1}_q$ if and only if the $i$-th vertex of $W$ is $v^r_q$.  Of course $W^+$ is well defined if and only if $W$ contains no vertices at row $20n$.  

Consider the family $\C$ of cycles whose vertical edges all have $\Wg_Y = y_\infty$.  It is not hard to see that, for every $C \in \C$, either $C' = C^+$ is well defined or $C = (C')^+$ for some cycle $C'$.  Whichever the case, $C' \in \C$ and, moreover, $C \cup C'$ has:
\begin{itemize}
 \item $20n-2$ vertical edges with $\Wg_Y = y_\infty$,
 \item one diagonal edge with $\Wg_Y = l_iT$ (resp.\ $\Wg_Y = h_iT$) for every $i \in \Range{n}$,
 \item $3n$ diagonal edges with $\Wg_Y = T$ and one diagonal edge with $\Wg_Y = s_i$ for $i \in \Range{n}$, and
 \item $12n$ edges with $\Wg_Y = 1$.
\end{itemize}
That is, $\Wg_Y(C) + \Wg_Y(C') = 2\Len$.  Since $(X,\Len)$ is a yes instance of MCW, this means that $\Wg_Y(C) = \Len$ for every $C \in \C$.

\begin{figure}
  \mbox{}\hfill \includegraphics{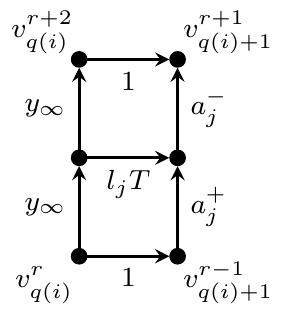} \hfill \includegraphics{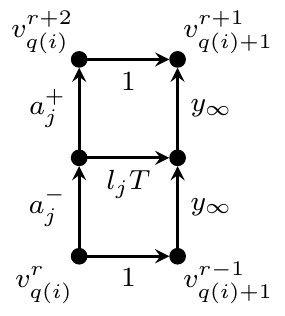} \hfill \includegraphics{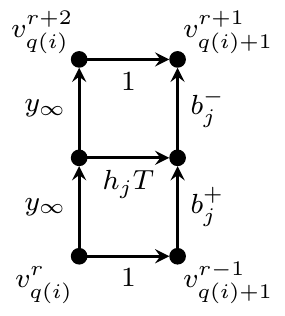} \hfill \includegraphics{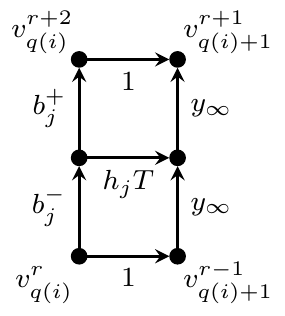} \hfill\mbox{}
  \caption{Section of $G$ for a column $q(i)$ whose vertical edges have $\Wg_Y = y_\infty$.  The figures on the left/right correspond to the cases in which $q(i)$ is even/odd.}\label{fig:heavy columns}
\end{figure}

Let $Q = q(0) < \ldots < q(2n-1) \subseteq \Range{10n}$ be the columns in which the vertical edges of $G$ have $\Wg_Y = y_\infty$ (\Figure~\ref{fig:heavy columns}).  Without loss of generality, we may assume that $L_0$ is the first tuple in $Y$, thus $q(0) = 0$.  Fix $i \in \Range{2n}$ and let $\alpha = q(i)$, $\beta = q(i+1)$ and $\delta = \beta - \alpha$.  Clearly, any path $D_i$ from $v_{\alpha}^{r+\delta}$ to $v_{\beta}^{r}$ has only diagonal edges.  It is easy to see that $P = D_i, v_{\beta}^{r+1}$ is a subpath of some cycle $C \in \C$.  Moreover, if we replace $P$ with $P' = v_{\alpha}^{r+\delta}, D_i^+$ in $C$, then we obtain another cycle $C' \in \C$.  Therefore, as the vertical edges at columns $\alpha$ and $\beta$ have $\Wg_Y = \infty$, it follows that $\Wg_Y(D_i) = \Wg_Y(D_i^+)$.  

Write $D_i[j]$ and $D^+_i[j]$ to denote the $j$-th edges of $D_i$ and $D^+_i$, respectively, for $j \in \Range{\delta}$.  By construction, either $\Wg_Y(D_i[0]) \leq T$ or $\Wg_Y(D^+_i[0]) \leq T$.  Suppose the former without loss of generality.  Then, if $\alpha$ is even, it follows that $\Wg_Y(D^+_i[0]) = l_jT + aT$ for some $j \in \Range{n}$ and $a \in \{0,2\}$, while $\Wg_Y(D_i[0]) = 1$ (\Figure~\ref{fig:heavy columns}).  Otherwise, if $\alpha$ is odd, then $\Wg_Y(D^+_i[0]) \leq T$ (\Figure~\ref{fig:heavy columns}).  Similarly, if $\beta$ is odd, then one of $D_i[\delta-1]$ and $D^+_i[\delta-1]$ has $\Wg_Y = 1$ while the other has $\Wg_Y = l_jT + aT$ for some $j \in \Range{n}$ and $a \in \{0,2\}$; otherwise, both $D_i[\delta-1]$ and $D^+_i[\delta-1]$ have $\Wg_Y \leq T$ (\Figure~\ref{fig:heavy columns}).  Finally, the $O(n)$ other edges of $D_i$ and $D^+_i$ have $\Wg_Y \leq T$.  Therefore, as $\Wg_Y(D_i) = \Wg_Y(D^+_i)$, it follows that $\alpha$ is even if and only if $\beta$ is odd.  Moreover, if $\alpha$ is even, then $\Wg_Y(D^+_i[0]) = l_jT + aT$, $\Wg_Y(D_i[\delta-1]) = l_jT + (2-a)T$, and $\Wg_Y(D^+_i[\delta - 1]) = \Wg_Y(D_i[0]) = 0$, for some $j \in \Range{n}$ and $a \in \{0,2\}$.  

The above facts imply that $\delta \geq 9$.  Then, taking into account that $G$ has $10n$ columns and $q(0) = 0$, it follows that $q(2i) = 10i$ and $q(2i+1) = 10i +9$, for every $i \in \Range{n}$.  Moreover, one between 
\begin{displaymath}
\left\{\vphantom{\Wg_Y(D^+[])}\Wg_Y(D_i[2]), \Wg_Y(D_i[4]), \Wg_Y(D_i[6])\right\} \text{ and } \left\{\Wg_Y(D^+_i[2]), \Wg_Y(D^+_i[4]), \Wg_Y(D^+_i[6])\right\}
\end{displaymath}
is a subset of $S$, called $S_i$, with $\sum S_i = T$.  Summing up, $S_0, \ldots, S_{n-1}$ certifies that $S$ is a yes instance of the $3$-partition problem.
\end{proof}

\begin{theorem}\label{thm:minimum NP}
 The minimum representation problem is \NP-complete.
\end{theorem}

\begin{proof}
 Viewed as a decision problem, the goal in the minimum representation problem is to find a $(\Circ^*,\Len^*,\Dist,\BegDist)$-CA model $\Unit$ isomorphic to a $(\Circ, \Len, \Dist, \BegDist)$-CA model $\M$ such that $\Circ^* \leq \Circ$ and $\Len^* \leq \Len$, when $\M$, $\Circ$, $\Len$, $\Dist$, and $\BegDist$ are given as input.  Clearly, this problem belongs to \NP, as we can take $\Unit$ and an isomorphism $f$ between $\M$ and $\Unit$ as the certificate.  To prove its hardness, we show a polynomial time reduction from MCW that is strongly \NP-complete by Theorem~\ref{thm:MCW is NPC}.  
 
 Let $(X, \Len)$ be an input of MCW, $k = |X|$, and $X_q = (x_q^0, x_q^1, y_q^0, y_q^1)$ be the $q$-th tuple in $X$ for $q \in \Range{k}$.  Call $\M$ to the PCA model depicted in \Figure~\ref{fig:np-models}, and let $\M_q$ be the PCA model obtained from $\M$ after inserting $(y_q^0-1)$ copies of $A_1$, $(x_q^1-1)$ copies of $A_2$, $(y_q^1-1)$ copies of $A_4$, and $(x_q^0-1)$ copies of $A_5$ ($\Syn(\M_q)$ is depicted in \Figure~\ref{fig:np-models}). Clearly, the PCA model $\M_X = \M_0 + \ldots + \M_{k-1}$ can be computed in polynomial time, provided that the numbers of $X$ are encoded in the unary system.  In the following we show that $(X,\Len)$ is a yes instance of MCW if and only if $(\M_X, \infty, \Len, 1, 0)$ is a yes instance of the (decision version of the) minimum representation problem.  For this, is enough to prove that:
 \begin{enumerate}[(i)]
  \item for every $Y = \pi(\chi(\gamma(X)))$ there exists a UCA model $\M(Y)$ isomorphic to $\M_X$,
  \item for every PCA model $\M$ isomorphic to $\M_X$ there exist $Y = \pi(\chi(\gamma(X)))$ such that $\M$ is shift equivalent to $\M(Y)$, and 
  \item the minimal $(\Circ,\Len)$-CA model equivalent to $\M(Y)$ has $\Len = \Wg_Y(G)$, for $Y = \pi(\chi(\gamma(X)))$.
 \end{enumerate} 
 It is important to remark that $\M(Y)$ has nothing to do with $\M_Y$.  The former is a model isomorphic to $\M_X$ that depends on $\pi$, $\chi$, and $\gamma$, while the latter denotes the reduction when the input is $Y$.  Thus, $\M_X$ and $\M_Y$ need not be isomorphic.
 
 To define $\M(Y)$ we transform every co-component $\M_q$ of $\M_X$, for $q \in \Range{k}$.  Specifically, let:
 \begin{itemize}
  \item $\M_q^\chi = (\M_q^{-1})|1$ if $\chi(q) = 1$ and $\M_q^\chi = \M_q$ otherwise, and
  \item $\M_q^{\chi,\gamma} = (\M_q^\chi)^{-1}$ if $\gamma(q) = 1$ and $\M_q^{\chi,\gamma} = \M_q^{\chi}$ otherwise.
 \end{itemize}
 Then, $\M(Y) = \pi(\M_0^{\chi,\gamma} + \ldots + \M_{k-1}^{\chi,\gamma})$.  Clearly, by taking different values for $\chi(q)$ and $\gamma(q)$, we can generate $\M_q^{j}|i$ for every $i \in \{0,1\}$ and $j \in \{-1,1\}$.  Therefore, (i)~and~(ii) follow by Theorem~\ref{thm:all models}.
 
 To prove (iii), note that $\M$ is equivalent to both $\M|1$ and $\M^{-1}$ (\Figure~\ref{fig:np-models}).  Then, as inserting copies of arcs is commutative with reverse and alignment operations, it follows that the synthetic graph $\S_q$ of $\M_q^{\chi,\gamma}$ is one of those depicted in \Figure~\ref{fig:np-models}, perhaps after a $1$-alignment is applied.  Therefore, if $\pi(q) = p$ and we write:
 \begin{itemize} 
  \item $A_q^r$ and $B_q^r$ as the leftmost and rightmost vertices at row $r$ of $(2k)\Unroll\Syn_q$, $r \in \Range{4k}$, and
  \item $\Wg_Y(v_{2p}^r, v_{2p+1}^h)$ to denote the maximum among the weights of the paths that go from $v_{2p}^r$ to $v_{2p+1}^h$ in $G$ for $r,h \in \Range{4k}$, 
 \end{itemize}
 then we immediately obtain that the length of the longest path from $A_q^r$ to $B_q^h$ in $(2k)\Unroll\Syn_q$ is precisely $\Wg_Y(v_{2p}^r, v_{2p+1}^h)$ (\Figures\ \ref{fig:swaps}~and~\ref{fig:np-models}).  Consequently, $\Wg_Y(G)$ is equal to the length of the longest cycle in $(2k)\Unroll\Syn(Y)$, where $\Syn(Y) = \Syn_{\pi^{-1}(0)} + \ldots + \Syn_{\pi^{-1}(k-1)}$ is the synthetic graph of $\M(Y)$.  Moreover, as $r \geq p+1$ for every forward path of $\Syn_q$ from $A_q^r$ to $B_q^p$, then any cycle of $(3n)\Unroll\Syn(Y)$ has a copy in $(2k)\Unroll\Syn(Y)$, where $n$ is the number of arcs in $\M(Y)$. Therefore (iii) follows by~\eqref{eq:minimal len}, as $\M$ is equivalent to a minimal $(\Circ, \Wg_Y(G))$-CA model.
 \end{proof}

\begin{figure}[t!]
  \centering
  \includegraphics{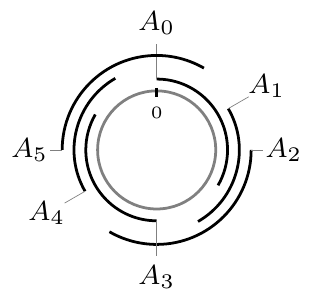} \includegraphics{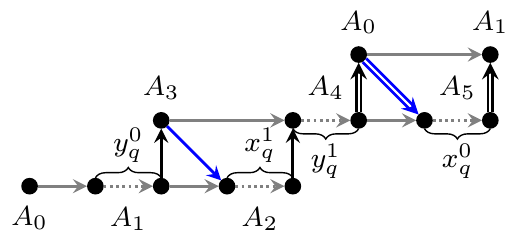} \includegraphics{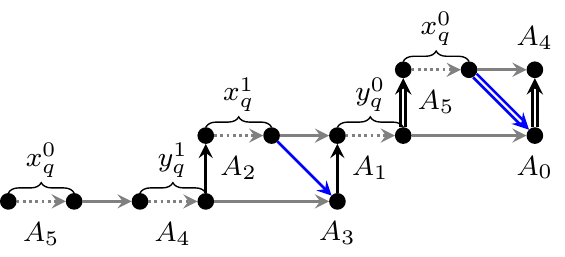} 
\caption{PCA model $\M$ of Theorem~\ref{thm:minimum NP} and Mitas' drawings, with external edges, of $\Syn(\M_q)$ and $\Syn(\M_q^{-1})$ for the PCA model $\M_q$ obtained after inserting copies of $A_1$, $A_2$, $A_4$, and $A_5$.}\label{fig:np-models}
\end{figure}

\section{Conclusions}

In this article we provided an improved algorithm for the minimal representation problem, while we proved that the minimum representation problem is \NP-complete.  A key contribution was to observe that the minimal length of a UCA model is determined by length of a maximum cycle in the synthetic graph obtained after a loop unrolling.  Loop unrolling is an old and simple technique born to improve the speed of computer programs.  It was already applied on the study coloring problems over circular-arc models, mainly for compiler design.  Here, instead, we use loop unrolling to understand the global structure of PCA and UCA models. We believe that the combination of loop unrolling with synthetic graphs provides a promising framework for further research.

We remark that even though many properties of UIG models extended naturally to UCA models, this is not always the case, as UCA models have a much richer structure than UIG models.  Indeed, most of the representation algorithms that generate a UIG model of an input graph do not extend to the circular case because Robert's PIG=UIG theorem does not hold in the circular case.  The fact that synthetic graphs behave so well in the circular case is a plus for this tool.  But, what is more surprising for us, is that we can translate the information in the circular structure into a linear one by unrolling $O(n)$ times the model.  

Finally, we mention that Pirlot's original definition of minimality is stronger than the one we discuss in this article.  Say that a $(\Circ,\Len,\Dist,\BegDist)$-CA model $\Unit$ with arcs $A_0 < \ldots < A_{n-1}$ is \emph{strongly} minimal when $\Unit$ is minimal and \textbf{globally} \emph{left justified}.  The latter means that $s(A_i) < s(A_i')$ for every $i \in \Range{n}$ and every $(\Circ',\Len',\Dist,\BegDist)$-CA model equivalent to $\Unit$ whose arcs are $A_1' < \ldots < A_n'$.  The minimal $\Desc$-CA model computed by Theorem~\ref{thm:algorithm} is \textbf{locally} left justified, as it satisfies the previous condition for those $\Desc$-CA models equivalent to $\Unit$~\cite{SoulignacJGAA2017,SoulignacJGAA2017a}.  It remains as an open problem to characterize when a UCA model is equivalent to a strongly UCA model.

\small
% \bibliographystyle{notabbrvnat}
% \bibliography{biblio}

\begin{thebibliography}{19}
\providecommand{\natexlab}[1]{#1}
\providecommand{\url}[1]{\texttt{#1}}
\expandafter\ifx\csname urlstyle\endcsname\relax
  \providecommand{\doi}[1]{doi: #1}\else
  \providecommand{\doi}{doi: \begingroup \urlstyle{rm}\Url}\fi

\bibitem[Costa et~al.(2012)Costa, Dantas, Sankoff, and
  Xu]{CostaDantasSankoffXuJBCS2012}
V.~Costa, S.~Dantas, D.~Sankoff, and X.~Xu.
\newblock Gene clusters as intersections of powers of paths.
\newblock \emph{J. Braz. Comput. Soc.}, 18\penalty0 (2):\penalty0 129--136,
  2012.
\newblock \doi{10.1007/s13173-012-0064-8}.

\bibitem[de~Werra et~al.(1999)de~Werra, Eisenbeis, Lelait, and
  Marmol]{WerraEisenbeisLelaitMarmolDAM1999}
D.~de~Werra, C.~Eisenbeis, S.~Lelait, and B.~Marmol.
\newblock On a graph-theoretical model for cyclic register allocation.
\newblock \emph{Discrete Appl. Math.}, 93\penalty0 (2-3):\penalty0 191--203,
  1999.
\newblock \doi{10.1016/S0166-218X(99)00105-5}.

\bibitem[Deng et~al.(1996)Deng, Hell, and Huang]{DengHellHuangSJC1996}
X.~Deng, P.~Hell, and J.~Huang.
\newblock Linear-time representation algorithms for proper circular-arc graphs
  and proper interval graphs.
\newblock \emph{SIAM J. Comput.}, 25\penalty0 (2):\penalty0 390--403, 1996.
\newblock \doi{10.1137/S0097539792269095}.

\bibitem[Dur{\'a}n et~al.(2006)Dur{\'a}n, Gravano, McConnell, Spinrad, and
  Tucker]{DuranGravanoMcConnellSpinradTuckerJA2006}
G.~Dur{\'a}n, A.~Gravano, R.~M. McConnell, J.~Spinrad, and A.~Tucker.
\newblock Polynomial time recognition of unit circular-arc graphs.
\newblock \emph{J. Algorithms}, 58\penalty0 (1):\penalty0 67--78, 2006.
\newblock \doi{10.1016/j.jalgor.2004.08.003}.

\bibitem[Dur{\'{a}}n et~al.(2015)Dur{\'{a}}n, Fern{\'{a}}ndez~Slezak, Grippo,
  de~Souza~Oliveira, and
  Szwarcfiter]{DuranFernandezGrippoSouzaSzwarcfiterENiDM2015}
G.~Dur{\'{a}}n, F.~Fern{\'{a}}ndez~Slezak, L.~N. Grippo, F.~de~Souza~Oliveira,
  and J.~L. Szwarcfiter.
\newblock On unit interval graphs with integer endpoints.
\newblock In \emph{L{AGOS}'15---\{VIII\} {L}atin-{A}merican {A}lgorithms,
  {G}raphs and {O}ptimization {S}ymposium}, vol.~50 of \emph{Electron. Notes
  Discrete Math.}, pp. 445--450. Elsevier Sci. B. V., Amsterdam, 2015.
\newblock \doi{10.1016/j.endm.2015.07.074}.

\bibitem[Garey and Johnson(1979)]{GareyJohnson1979}
M.~R. Garey and D.~S. Johnson.
\newblock \emph{Computers and intractability}.
\newblock W. H. Freeman and Co., San Francisco, Calif., 1979.

\bibitem[Huang(1995)]{HuangJCTSB1995}
J.~Huang.
\newblock On the structure of local tournaments.
\newblock \emph{J. Combin. Theory Ser. B}, 63\penalty0 (2):\penalty0 200--221,
  1995.
\newblock \doi{10.1006/jctb.1995.1016}.

\bibitem[Kaplan and Nussbaum(2009)]{KaplanNussbaumDAM2009}
H.~Kaplan and Y.~Nussbaum.
\newblock Certifying algorithms for recognizing proper circular-arc graphs and
  unit circular-arc graphs.
\newblock \emph{Discrete Appl. Math.}, 157\penalty0 (15):\penalty0 3216--3230,
  2009.
\newblock \doi{10.1016/j.dam.2009.07.002}.

\bibitem[Klav{\'{\i}}k et~al.(2017)Klav{\'{\i}}k, Kratochv{\'{\i}}l, Otachi,
  Rutter, Saitoh, Saumell, and
  Vysko{\v{c}}il]{KlavikKratochvilOtachiRutterSaitohSaumellVyskocilA2017}
P.~Klav{\'{\i}}k, J.~Kratochv{\'{\i}}l, Y.~Otachi, I.~Rutter, T.~Saitoh,
  M.~Saumell, and T.~Vysko{\v{c}}il.
\newblock Extending partial representations of proper and unit interval graphs.
\newblock \emph{Algorithmica}, 77\penalty0 (4):\penalty0 1071--1104, 2017.
\newblock \doi{10.1007/s00453-016-0133-z}.

\bibitem[K\"obler et~al.(2017)K\"obler, Kuhnert, and
  Verbitsky]{KoeblerKuhnertVerbitskyDAM2017}
J.~K\"obler, S.~Kuhnert, and O.~Verbitsky.
\newblock Circular-arc hypergraphs: rigidity via connectedness.
\newblock \emph{Discrete Appl. Math.}, 217\penalty0 (part 2):\penalty0
  220--228, 2017.
\newblock \doi{10.1016/j.dam.2016.08.008}.

\bibitem[Lin and Szwarcfiter(2008)]{LinSzwarcfiterSJDM2008}
M.~C. Lin and J.~L. Szwarcfiter.
\newblock Unit circular-arc graph representations and feasible circulations.
\newblock \emph{SIAM J. Discrete Math.}, 22\penalty0 (1):\penalty0 409--423,
  2008.
\newblock \doi{10.1137/060650805}.

\bibitem[Lin et~al.(2009)Lin, Soulignac, and
  Szwarcfiter]{LinSoulignacSzwarcfiter2009}
M.~C. Lin, F.~J. Soulignac, and J.~L. Szwarcfiter.
\newblock Short models for unit interval graphs.
\newblock In \emph{L{AGOS}'09---{V} {L}atin-{A}merican {A}lgorithms, {G}raphs
  and {O}ptimization {S}ymposium}, vol.~35 of \emph{Electron. Notes Discrete
  Math.}, pp. 247--255. Elsevier Sci. B. V., Amsterdam, 2009.
\newblock \doi{10.1016/j.endm.2009.11.041}.

\bibitem[Mitas(1994)]{Mitas1994}
J.~Mitas.
\newblock Minimal representation of semiorders with intervals of same length.
\newblock In \emph{Orders, algorithms, and applications ({L}yon, 1994)}, vol.
  831 of \emph{Lecture Notes in Comput. Sci.}, pp. 162--175. Springer, Berlin,
  1994.
\newblock \doi{10.1007/BFb0019433}.

\bibitem[Pirlot(1990)]{PirlotTaD1990}
M.~Pirlot.
\newblock Minimal representation of a semiorder.
\newblock \emph{Theory and Decision}, 28\penalty0 (2):\penalty0 109--141, 1990.
\newblock \doi{10.1007/BF00160932}.

\bibitem[Pirlot and Vincke(1997)]{PirlotVincke1997}
M.~Pirlot and P.~Vincke.
\newblock \emph{Semiorders}, vol.~36 of \emph{Theory and Decision Library.
  Series B: Mathematical and Statistical Methods}.
\newblock Kluwer Academic Publishers Group, Dordrecht, 1997.
\newblock \doi{10.1007/978-94-015-8883-6}.

\bibitem[Soulignac(2015)]{SoulignacA2015}
F.~J. Soulignac.
\newblock Fully dynamic recognition of proper circular-arc graphs.
\newblock \emph{Algorithmica}, 71\penalty0 (4):\penalty0 904--968, 2015.
\newblock \doi{10.1007/s00453-013-9835-7}.

\bibitem[{Soulignac}(2017{\natexlab{a}})]{SoulignacJGAA2017}
F.~J. {Soulignac}.
\newblock Bounded, minimal, and short representations of unit interval and unit
  circular-arc graphs. {C}hapter {I}: theory.
\newblock \emph{J. Graph Algorithms Appl.}, 21\penalty0 (4):\penalty0 455--489,
  2017{\natexlab{a}}.
\newblock \doi{10.7155/jgaa.00425}.

\bibitem[{Soulignac}(2017{\natexlab{b}})]{SoulignacJGAA2017a}
F.~J. {Soulignac}.
\newblock Bounded, minimal, and short representations of unit interval and unit
  circular-arc graphs. {C}hapter {II}: algorithms.
\newblock \emph{J. Graph Algorithms Appl.}, 21\penalty0 (4):\penalty0 491--525,
  2017{\natexlab{b}}.
\newblock \doi{10.7155/jgaa.00426}.

\bibitem[Tucker(1974)]{TuckerDM1974}
A.~Tucker.
\newblock Structure theorems for some circular-arc graphs.
\newblock \emph{Discrete Math.}, 7:\penalty0 167--195, 1974.
\newblock \doi{10.1016/S0012-365X(74)80027-0}.

\end{thebibliography}

\end{document}